\documentclass[a4paper,onecolumn,11pt,accepted=2026-03-04]{quantumarticle}
\pdfoutput=1

\usepackage[utf8]{inputenc}
\usepackage[T1]{fontenc}
\usepackage{hyperref}

\usepackage{libpapers}
\usepackage{libgraphics}

\usepackage[numbers,compress]{natbib}

\usepackage{latexsym}

\usepackage{algorithm}
\usepackage{algpseudocode}
\algrenewcommand\algorithmicrequire{\textbf{Input:}}
\algrenewcommand\algorithmicensure{\textbf{Output:}}

\usepackage{quantikz}

\newcommand{\gammatwo}{{\mathop{\gamma_2}\limits^\leftarrow}}
\newcommand{\gammatwodual}{{\mathop{\gamma_2}\limits^{\mathbin{\rotatebox[origin=c]{180}{$\curvearrowright$}}}}}
\newcommand{\gammatwoplus}{\gammatwo_+}
\newcommand{\gammatwodualplus}{\gammatwodual_+}

\numberwithin{theorem}{section}

\begin{document}

\title{An adversary bound for quantum signal processing}
\author{Lorenzo Laneve}
\affiliation{Faculty of Informatics — Universit\`a della Svizzera Italiana, 6900 Lugano, Switzerland}
\email{lorenzo.laneve1@gmail.com}
\orcid{0000-0003-2319-5456}
\thanks{}
\maketitle

\begin{abstract}
  Quantum signal processing (QSP) and quantum singular value transformation (QSVT), have emerged as unifying frameworks in the context of quantum algorithm design. These techniques allow to carry out efficient polynomial transformations of matrices block-encoded in unitaries, involving a single ancilla qubit. Recent efforts try to extend QSP to the multivariate setting (M-QSP), where multiple matrices are transformed simultaneously. However, this generalization faces problems not encountered in the univariate counterpart: in particular, the class of polynomials achievable by M-QSP seems hard to characterize. In this work we borrow tools from query complexity, namely the state conversion problem and the adversary bound: we first recast QSP as a state conversion problem over the Hilbert space of square-integrable functions. We then show that the adversary bound for a state conversion problem in this space precisely identifies all and only the QSP protocols in the univariate case. Motivated by this first result, we extend the formalism to several variables: {the existence of a feasible solution to the adversary bound implies the existence of a M-QSP protocol, and the computation of a protocol of minimal space is reduced to a rank minimization problem involving the feasible solution space of the adversary bound.}
\end{abstract}

\tableofcontents

\pagebreak
\section{Introduction}
One of the main tasks in quantum computation involves quantum algorithm design: given some computational problem, how can we devise procedures that offer provable speed-up compared to classical models? To solve this task effectively, we need a set of generally applicable techniques that we can combine to construct useful algorithms. For this purpose, many different approaches were proposed throughout the years including amplitude amplification schemes~\cite{brassardQuantumAmplitudeAmplification2002}, linear combination of unitaries~\cite{berryHamiltonianSimulationNearly2015,childsHamiltonianSimulationUsing2012,anLinearCombinationHamiltonian2023,anLaplaceTransformBased2024}, phase estimation~\cite{nielsenQuantumComputationQuantum2010a}, matrix inversion~\cite{harrowQuantumAlgorithmLinear2009} and quantum walks~\cite{aharonovQuantumWalksGraphs2001,ambainisQuantumWalkAlgorithm2004,szegedyQuantumSpeedMarkovChain2004,belovsQuantumWalksElectric2013,apersUnifiedFrameworkQuantum2021}.

More recently, a technique called \emph{quantum signal processing} (QSP) was developed which --- along with its lift to the \emph{quantum singular value transformation} (QSVT)~\cite{gilyenQuantumSingularValue2019} --- was proven to unify and simplify a vast number of the aforementioned techniques and beyond~\cite{lowHamiltonianSimulationUniform2017,lowOptimalHamiltonianSimulation2017,lowQuantumSignalProcessing2017,lowHamiltonianSimulationQubitization2019,martynGrandUnificationQuantum2021,berryDoublingEfficiencyHamiltonian2024}. QSP/QSVT involves the application of a polynomial transformation to matrices that are block-encoded in a unitary. The major advantage of this ansatz is that it requires a single ancilla qubit to work and the class of implementable polynomials gives great freedom. Moreover, numerical algorithms for the computation of QSP protocols starting from a desired transformation are efficient and well-understood~\cite{haahProductDecompositionPeriodic2019,chaoFindingAnglesQuantum2020,dongEfficientPhasefactorEvaluation2021,yingStableFactorizationPhase2022,wangEnergyLandscapeSymmetric2022,dongRobustIterativeMethod2024,yamamotoRobustAngleFinding2024,berntsonComplementaryPolynomialsQuantum2025,alexisInfiniteQuantumSignal2024,laneveGeneralizedQuantumSignal2025,niFastPhaseFactor2024,niInverseNonlinearFast2025}. Notable mentions across the QSP literature include randomized QSP~\cite{martynHalvingCostQuantum2025}, parallel QSP~\cite{martynParallelQuantumSignal2024}, recursive QSP~\cite{mizutaRecursiveQuantumEigenvalue2024}, composable QSP~\cite{rossiModularQuantumSignal2025} and protocol synthesis through a Solovay-Kitaev-like approximation~\cite{rossiSolovayKitaevTheoremQuantum2025}.

Given the success of QSP, which involves the eigenvalue or singular value transformation of a matrix, recent research tries to generalize the ansatz to multiple signals (we call this \emph{multivariate quantum signal processing}, or M-QSP~\cite{rossiMultivariableQuantumSignal2022,moriCommentMultivariableQuantum2024,itoPolynomialTimeConstructive2024}), which in turn allows the joint transformation of multiple matrices. Unfortunately, the flexibility of univariate QSP --- in terms of achievable polynomials --- does not carry over to more than one variable, and some impossibility results were found~\cite{nemethVariantsMultivariateQuantum2023,laneveMultivariatePolynomialsAchievable2025}. Some approaches combine QSP with linear combination of unitaries~\cite{borns-weilQuantumAlgorithmFunctions2023}, giving up the single-qubit structure, while others go for a modular, bottom-up approach~\cite{rossiModularQuantumSignal2025,gomesMultivariableQSPBosonic2024}, which requires a bit of engineering to achieve a desired transformation --- as opposed to the original top-down approach, where the target polynomial can be algorithmically compiled into an implementing protocol.

On the other hand, the theory of quantum query complexity is more mature and well-understood: the query complexity, i.e., the number of calls to an oracle needed for an algorithm to solve a certain task --- was studied in terms of lower bounds. A first example of the application of this theory is to prove that Grover's search is optimal~\cite{hoyerLowerBoundsQuantum2005}. Two main techniques to derive lower bounds for quantum query complexity include the polynomial method~\cite{bealsQuantumLowerBounds2001,ambainisPolynomialDegreeVs2004} and the adversary bound~\cite{ambainisQuantumLowerBounds2000,spalekMultiplicativeQuantumAdversary2008,spalekAllQuantumAdversary2005,belovsVariationsQuantumAdversary2015,hoyerTightAdversaryBounds2006}. While the former was already used in the context of lower bounds for matrix functions, where QSP/QSVT are protagonists~\cite{montanaroQuantumClassicalQuery2024}, the latter was proven to be stronger {in the bounded error setting}~\cite{hoyerNegativeWeightsMake2007,reichardtSpanProgramsQuantum2009,belovsDirectReductionPolynomial2024}: in particular, the adversary bound not only constitutes a lower bound, but also gives an upper bound, i.e., from a feasible solution to the adversary bound, one can construct a quantum query algorithm. A central task in the theory of query complexity is \emph{state conversion}~\cite{ambainisSymmetryassistedAdversariesQuantum2011,leeQuantumQueryComplexity2011,belovsOneWayTicketVegas2023}, where we are given one copy of some state $\ket{\xi_x}$ which needs to be converted into some final state $\ket{\tau_x}$, extracting information on $x$ only from the oracle.

In this work we show that (M-)QSP can be seen as an instance of state conversion {over the Hilbert space of square-integrable functions}. In particular, the adversary bound --- which already characterizes the original state conversion problem --- also characterizes QSP: in the case of univariate single-qubit QSP, we find a bijection between the feasible solutions of the adversary bound and the $SU(2)$-QSP protocols, which motivates the use of this formalism to explore M-QSP. Moreover, we find that, for a given desired polynomial transformation $P$ (either univariate or multivariate), the feasible solution space of the adversary bound is a convex set of all the possible protocols implementing $P$, for which we ideally want to find a minimal-rank solution to optimize space. {These results imply that the properties of QSP protocols can be studied by looking directly at the adversary bound. In particular for the multivariate case a feasible solution to the adversary bound \emph{certifies} the existence of a M-QSP protocol.}

The remainder of this work is structured as follows: Section~\ref{sec:finite-state-conversion} is a review of the theory of state conversion over finite label sets, mainly taken from~\cite{leeQuantumQueryComplexity2011,belovsOneWayTicketVegas2023}, introduces a new variant called \emph{partial state conversion}, {and defines the original adversary bound}. Section~\ref{sec:qsp} gives an overview of QSP, defining QSP protocols in both univariate and multivariate case. Section~\ref{sec:state-conversion-l2}, {containing our main technical results,} is devoted to the merge of these two theories, defining \emph{polynomial state conversion} and characterizing the structure of the feasible solution space of the corresponding adversary bound. Section~\ref{sec:discussion} concludes with remarks and outlook for future directions. The full technical proofs for Section~\ref{sec:state-conversion-l2} can be found in Appendix~\ref{apx:catalyst-proofs}, {and a quick introduction to the theory of linear operators over Hilbert spaces can be found in Appendix~\ref{apx:trace-class}.}

\subsection{Preliminaries and notation}
Given vectors $v_1, v_2, \ldots, v_n$, the direct sum can be defined as
\begin{align*}
    v_1 \oplus \cdots \oplus v_n := {\sum_{k = 1}^n} \ket{k} \ket{v_k} \ .
\end{align*}
The vectors might be of different lengths, and in this case the second register is taken large enough so that it can contain all the $v_k$'s. A polynomial vector is a vector whose entries are polynomials or, equivalently, a polynomial whose coefficients are vectors in some Hilbert space
\begin{align*}
    v(z) = \sum_{j} P_j(z) \ket{j} = \sum_k v_k z^k \ .
\end{align*}
The vector space $\vspan{ v(z) }_z$, which coincides with $\vspan{ v_k }_k$ is the \emph{range} of $v(z)$. {We refer to the dimension of the range of $v(z)$ as its \emph{dimension} or \emph{rank}.} In particular, if $v(z)$ has rank $d$, then there exists a unitary $U$ such that the range of $U v(z)$ is {$\vspan{\ket{0}, \ldots, \ket{d-1}}$}. The \emph{support} of $v(z)$ is the set of $k$ for which $v_k \neq 0$, while the \emph{degree} is the maximum such $k$. These definitions can be extended to square-integrable functions $v(z)$, for which the support will not necessarily be finite. Moreover, the definition can be extended to multiple variables, where we can define the general degree --- which is the degree of the maximum monomial appearing in the sum --- and the degree with respect to each variable, following the univariate definition. Here \emph{degree} will refer to the former unless explicitly said otherwise. A \emph{polynomial state} $\tau(z)$ is a polynomial vector satisfying $\norm{\tau(z)} = 1$ for any choice of $z \in \T$ ($\T$ being the complex unit circle), and analogously for the multivariate case on $\T^m$. If the condition relaxes to $\norm{\tau(z)} \le 1$, then we refer to this object as a \emph{partial state}.

For some set of vectors $v = \{ v_x \}_{x \in D}$, we define the Gram matrix $G_v$ as the $D \times D$ matrix, whose entries are $(G_v)_{xy} = \langle v_x, v_y \rangle$. It is a well-known fact in linear algebra that two sets of vectors $v, v'$ have the same Gram matrix if and only if there exists a unitary transformation $U$ such that $U v_x = v'_x$ for every $x \in D$, and we will make extensive use of this fact.

\section{Finite state conversion and the adversary bound}
\label{sec:finite-state-conversion}

\begin{figure}
    \begin{center}
        \input{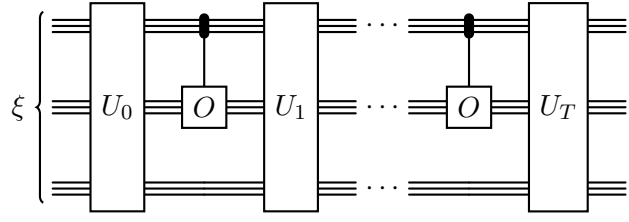}
        \caption{Circuit for a generic quantum query algorithm. The oracle $O$ can be controlled by some qubits, and does not need to act on all of them.}
        \label{fig:query-algorithm}
    \end{center}
\end{figure}

In this section we review the theory of state conversion, a central problem in query complexity. {We start by introducing the notion of query complexity in~\ref{sec:query-algorithms}, mainly following~\cite{belovsOneWayTicketVegas2023} (which we recommend for a more comprehensive treatment). We proceed in~\ref{sec:partial-state-conversion},~\ref{sec:gamma-bounds} by defining \emph{partial state conversion}, and the so-called $\gamma_2$-bounds. We use these semidefinite programs to define the adversary bound in~\ref{sec:adversary-bound-definition}, and we conclude in~\ref{sec:query-complexity-bounds} by proving that the feasible solutions of the adversary bound can be used to both provide a lower bound to the query complexity and construct algorithms to solve the state conversion problem.}

{\subsection{Quantum query algorithms and the state conversion problem}
\label{sec:query-algorithms}}
A \emph{quantum query algorithm} has access to an oracle, which is modeled as a linear operator (typically a unitary) $O : \calM \rightarrow \calM$. Without loss of generality, any such algorithm has the form (see Figure~\ref{fig:query-algorithm})
\begin{align*}
    \calA(O) = U_T \Tilde{O} U_{T-1} \Tilde{O} \cdots \Tilde{O} U_0
\end{align*}
where $U_k$ are unitary operators independent of $O$ acting on a Hilbert space $\calH = (\calM \otimes \calW) \oplus \calN$, while
\begin{equation}
    \label{eq:query-oracle-embedding}
    \Tilde{O} = (O \otimes \id_{\calW}) \oplus \id_{\calN}
\end{equation}
applies the oracle to some subspace of the total working space of the algorithm (this allows the procedure to decide whether to apply the oracle and in which subsystem). The most traditional definition of query complexity is the \emph{Monte Carlo query complexity}, which is simply the number $T$ of calls to $O$, while a second definition was introduced in~\cite{belovsOneWayTicketVegas2023}.
\begin{definition}
    \label{def:las-vegas-complexity}
    Given a starting state $\xi$ and an oracle $O$, the \emph{Las Vegas query complexity} of $\calA(O)$ is given by
    \begin{align*}
        L(\calA, O, \xi) = \sum_{t = 1}^T \norm{ \Pi U_t \Tilde{O} U_{t-1} \Tilde{O} \cdots \Tilde{O} U_0 \xi}^2
    \end{align*}
    where $\Pi$ is the projector onto {$\calM \otimes \calW$, the subspace where the controlled oracle $O \otimes \id_\calW$ is applied}, as in~(\ref{eq:query-oracle-embedding}).
\end{definition}
\noindent In other words, we take the probability mass of the component of the state that gets processed by the oracle at each step, {so this metric should be regarded as an \emph{expected cost} unlike the Monte Carlo definition}. This aligns with the classical case: the expected number of queries made by a randomized algorithm is $\sum_t p_t$, where $p_t$ is the probability that the $t$-th call is made, by linearity of expectation. 

We remark that, unlike Monte Carlo complexity, the Las Vegas definition depends on the starting state. The latter, however, satisfies an \emph{exactness} property~\cite{belovsTamingQuantumTime2024}, which essentially states that compositions of algorithms result in compositions of the respective Las Vegas complexities, without any log-factor penalty incurred for probability amplification. Moreover, the Las Vegas complexity introduced by the subroutines will depend on the distribution of the inputs they receive, which in certain cases might result in better total complexity (this is the \emph{thriftiness} property~\cite{belovsTamingQuantumTime2024}).

Using the setting defined so far, we look at a query complexity problem called \emph{state conversion}~\cite{leeQuantumQueryComplexity2011,belovsOneWayTicketVegas2023,belovsVariationsQuantumAdversary2015}. In this setting, oracles are indexed by a label $x \in D$ coming from a label set $D$ (which, for this section, will be finite), and we use $O_x$ to denote the unitary\footnote{Belov and Yolcu~\cite{belovsOneWayTicketVegas2023} also treat the case where $O_x$ are general contractions instead of unitaries, but we will only consider the unitary case here.} oracle when the input is $x$.

In a state conversion instance, we are given two sets of states $\{ \xi_x \}_{x \in D}, \{ \tau_x \}_{x \in D}$, and the goal is to construct a quantum query algorithm that satisfies, for every $x \in D$
\begin{align}
    \label{eq:state-conversion-output-condition}
    \calA(O_x) \xi_x = \tau_x
\end{align}
Here we are hiding possible states on ancilla systems, that must be fixed and independent of $x$ (there is also a \emph{non-coherent} version of state conversion~\cite{leeQuantumQueryComplexity2011}, which allows the final state to have a garbage state $g_x$ dependent on $x$ on the ancilla). In \emph{$\epsilon$-approximate} state conversion, the output condition of~(\ref{eq:state-conversion-output-condition}) relaxes to
\begin{align}
    \label{eq:state-conversion-output-condition-approx}
    \norm{\calA(O_x) \xi_x - \tau_x} \le \epsilon \ .
\end{align}
State conversion is a central task because several important problems are special cases. For example, when the starting state $\xi_x \equiv \xi$ is fixed for every $x \in D$, we obtain \emph{state generation}~\cite{ambainisSymmetryassistedAdversariesQuantum2011}. Furthermore, if we set $\tau_x = \ket{f(x)}$ for some discrete function $f : D \rightarrow [d]$, $d \in \N$, we obtain the \emph{function evaluation} problem, one of the first problems for which quantum query lower bounds were studied~\cite{bealsQuantumLowerBounds2001,ambainisQuantumLowerBounds2000,hoyerLowerBoundsQuantum2005}.

In the context of state conversion, the relationship between Las Vegas and Monte Carlo complexity is well-understood: for an algorithm $\calA$ solving state conversion we define the worst-case Las Vegas complexity as
\begin{align*}
    \calL(\calA) = \max_{x \in D} L(\calA, O_x, \xi_x)
\end{align*}
\begin{theorem}[\cite{belovsOneWayTicketVegas2023}]
    Given an algorithm $\calA$ for exact $\xi \mapsto \tau$ conversion there exists an algorithm $\calA'$ for $\epsilon$-approximate $\xi \mapsto \tau$ state conversion with Monte Carlo complexity $\bigO(\calL(\calA)/\epsilon^2)$.
\end{theorem}
\noindent This is analogous to what happens in the classical case: any Las Vegas algorithm with expected complexity $T$ can be turned into a Monte Carlo algorithm with failure probability $\le p$ and complexity $T/p$, by a standard application of Markov's inequality. Here we simply have $p = \epsilon^2$, by~(\ref{eq:state-conversion-output-condition-approx}).

\subsection{Partial state conversion}
\label{sec:partial-state-conversion}
We introduce a new version of state conversion, called \emph{partial state conversion}. In this variant, the output condition in (\ref{eq:state-conversion-output-condition-approx}) becomes
\begin{align}
    \label{eq:partial-state-conversion-output-condition-approx}
    \norm{\Pi \calA(O_x) \xi_x - \tau_x} \le \epsilon
\end{align}
for some orthogonal projection $\Pi$ independent of $x$. In other words, we only care about the behavior of the final state in a fixed subspace of $\calH$. This is motivated by our objective, since in the context of the block-encoding framework~\cite{gilyenQuantumSingularValue2019} or quantum signal processing~\cite{lowHamiltonianSimulationQubitization2019,martynGrandUnificationQuantum2021}, we usually care about the behavior of $\tau$ only in a portion of the space, implying that the residual can be chosen freely to optimize the complexity measures. In other words, given $\xi, \tau$ satisfying $\norm{\xi_x} \ge \norm{\tau_x}$, we want an algorithm that converts $\xi_x$ into $\tau_x \oplus \sigma_x$, for some complementary state $\sigma_x$ satisfying $\norm{\xi_x}^2 = \norm{\tau_x}^2 + \norm{\sigma_x}^2$, and the second Hilbert space will be projected out by $\Pi$.

Notice that partial state conversion is fundamentally different from the contraction oracle setting provided in~\cite{belovsOneWayTicketVegas2023}, since the contraction necessarily implies a permanent loss of probability mass (dependent on the input) at each call of the oracle.

\subsection{The $\gamma_2$-bound and its properties}
\label{sec:gamma-bounds}
The $\gamma_2$-bounds are semidefinite programs that are central in the query complexity setting. The first $\gamma_2$-bound in the context of state conversion was introduced in~\cite{leeQuantumQueryComplexity2011}, then further refined in~\cite{belovsOneWayTicketVegas2023}. As we shall review, the feasible solutions to the $\gamma_2$-bounds correspond precisely to the quantum query algorithms for state conversion. These are nothing more than an extension of the generalized adversary bound, which characterizes the query complexity of function evaluation problems~\cite{hoyerNegativeWeightsMake2007}.

{In the following, a $D \times D$ matrix for a set $D$ is intended to be a $|D| \times |D|$ matrix whose entries are indexed by elements of $D$.}

\begin{definition}[Unidirectional relative $\gamma_2$-bound~\cite{belovsOneWayTicketVegas2023}]
    \label{def:unidir-gamma-bound}
    Fix a $D \times D$ Hermitian matrix $E = \qty{ e_{xy} }$ and a $D \times D$ Hermitian block matrix $\Delta = \qty{ \Delta_{xy} }$ where each block is $\Delta_{xy}: \calM \rightarrow \calM$ for some Hilbert space $\calM$. Define $\gammatwo(E \,|\, \Delta)$ as the optimal value of the following semidefinite program:
    \begin{align}
        \text{minimize} \ \ \ \ & \max_{x \in D} \ \lVert v_x \rVert^2\nonumber \\
        \text{subject to} \ \ \ \ & e_{xy} = \langle v_x, (\Delta_{xy} \otimes \id_{\calW}) v_y \rangle \ \ \ \ \ \text{\text{for every} $x, y \in D$} \label{eq:unidir-gamma-bound-constraint} \\
        & \text{$\calW$ is a vector space, $v_x \in \calM \otimes \calW$}\nonumber
    \end{align}
\end{definition}
\noindent The $\Delta_{xy}$ blocks are called filter matrices and, by the fact that $\Delta$ is Hermitian, we must have $\Delta_{xy} = \Delta_{yx}^*$. {By defining the matrix $V = \sum_{x \in D} \ketbra{v_x}{x}$}, we can rewrite the constraint (\ref{eq:unidir-gamma-bound-constraint}) as a single matrix equation
\begin{align}
    \label{eq:unidir-gamma-bound-constraint-matrix}
    E = V^* (\Delta \otimes \id_{\calW}) V
\end{align}
We then define a partial version of this bound, which will be useful later.
\begin{definition}[Unidirectional partial relative $\gamma_2$-bound]
    \label{def:unidir-gamma-bound-partial}
    Fix a $D \times D$ Hermitian matrix $E = \qty{ e_{xy} }$ and a $D \times D$ Hermitian block matrix $\Delta = \qty{ \Delta_{xy} }$ where each block is $\Delta_{xy}: \calM \rightarrow \calM$ for some Hilbert space $\calM$. Define $\gammatwoplus(E \,|\, \Delta)$ as the optimal value of the following semidefinite program:
    \begin{align}
        \text{minimize} \ \ \ \ & \max_{x \in D} \ \lVert v_x \rVert^2\nonumber \\
        \text{subject to} \ \ \ \ & E \succeq V^* (\Delta \otimes \id_{\calW}) V \label{eq:unidir-gamma-bound-constraint-partial} \\
        & \text{$\calW$ is a vector space, $v_x \in \calM \otimes \calW$}\nonumber
    \end{align}
\end{definition}
\noindent The only difference introduced in the partial $\gamma_2$-bound is that the equality sign in (\ref{eq:unidir-gamma-bound-constraint-matrix}) is replaced with a semidefinite inequality. {To simplify the exposition in some points, we will sometimes use a slight abuse of notation and express semidefinite constraints directly with the matrix entries. For example, we write~(\ref{eq:unidir-gamma-bound-constraint-partial}) as
\begin{align*}
    e_{xy} \succeq \langle v_x, (\Delta_{xy} \otimes \id_{\calW}) v_y \rangle \ .
\end{align*}}
As with any semidefinite program, one can define a dual problem.
\begin{definition}[Unidirectional subrelative $\gamma_2$-bound~\cite{belovsOneWayTicketVegas2023}]
    \label{def:subrelative-gamma-bound}
    The relative $\gamma_2$-bound $\gammatwo(E \, |\, \Delta)$ has a dual formulation $\gammatwodual(E \, |\, \Delta)$ defined as
    \begin{align*}
        \text{maximize} \ \ \ \ & \lambda_{\max} (\Gamma \circ E) \\
        \text{subject to} \ \ \ \ & \lambda_{\max} (\Gamma \circ \Delta) \le 1
    \end{align*}
    where the optimization is over the $D \times D$ Hermitian matrices $\Gamma$, $\circ$ is the Hadamard (block-wise) product $(\Gamma \circ \Delta)_{xy} = \Gamma_{xy} \Delta_{xy}$, and $\lambda_{\max}(M)$ denotes the largest eigenvalue of $M$.
\end{definition}
Similarly, we define the partial version $\gammatwodualplus$ by adding the constraint $\Gamma \preceq 0$ (a full derivation for the total version is given in~\cite[Appendix~A]{belovsOneWayTicketVegas2023}, whereas the partial version comes by adding $\Lambda \preceq 0$ to their Lagrange multiplier $\Lambda$).

\begin{theorem}[Theorem 6.3 in~\cite{belovsOneWayTicketVegas2023}]
    \label{thm:relative-gamma-bound-weak-duality}
    Weak duality holds, i.e., $\gammatwodual \le \gammatwo$ and $\gammatwodualplus \le \gammatwoplus$
\end{theorem}
\begin{proof}
    Considering an optimal solution {$\{ v_x \}_{x \in D}$} for $\gammatwoplus(E \,|\, \Delta)$\footnote{The optimum might not be attained, but in this case we can replace $V$ with some feasible solution $V'$ whose objective value is $\gamma + \epsilon$, for any $\epsilon > 0$.}, we obtain, for any matrix $\Gamma \preceq 0$:
    \begin{equation}
        \label{eq:weak-duality-argument-inequality}
        \Gamma \circ E \preceq V^* \qty[(\Gamma \circ \Delta) \otimes \id_{\calW}] V
    \end{equation}
    We can then bound the maximal eigenvalue of $\Gamma \circ E$.
    \begin{align*}
        \lambda_{\max}(\Gamma \circ E) & = \max_{s : \norm{s} = 1} \langle s, (\Gamma \circ E) s\rangle \\
        & \le \max_{s : \norm{s} = 1} \langle Vs, \qty[(\Gamma \circ \Delta) \otimes \id_{\calW}] V s\rangle & \text{by (\ref{eq:weak-duality-argument-inequality})} \\
        & \le \max_{x \in D} \norm{v_x}^2 \cdot \max_{s' : \norm{s'} = 1} \langle s', \qty[(\Gamma \circ \Delta) \otimes \id_{\calW}] s'\rangle \\
        & = \gammatwoplus(E \,|\, \Delta) \cdot \lambda_{\max}(\Gamma \circ \Delta) \ .
    \end{align*}
    {For the second inequality, notice that applying $Vs$ is the same as applying $s' \propto Vs$ renormalized. The normalization factor would be $\norm{Vs} \le \norm{V} = \max_x \norm{v_x}$.}
    The argument for the total bound is obtained by replacing $\preceq$ with $=$ in~(\ref{eq:weak-duality-argument-inequality}).
\end{proof}
\noindent In our case, also strong duality holds (i.e., $\gammatwodual = \gammatwo$): for the total case, the dual formulation always satisfies Slater's condition, by taking $\Gamma = 0$~\cite{belovsVariationsQuantumAdversary2015,belovsOneWayTicketVegas2023}. The partial case is slightly subtler, as we additionally need $\Gamma \prec 0$, but we can simply take $\Gamma = - \epsilon \id$ for $\epsilon > 0$. This makes $\Gamma \circ \Delta$ block-diagonal, whose entries can be made arbitrarily small, and thus, strictly less than $1$ {in operator norm}. Constraint qualifications for the dual problems additionally guarantee that the optimum is attained.

\subsection{The adversary bound for state conversion}
\label{sec:adversary-bound-definition}
Consider now a state conversion problem $\xi_x \mapsto \tau_x$ as defined in Section~\ref{sec:finite-state-conversion}.

\begin{definition}
    We define the \emph{adversary bound} for the $\xi_x \mapsto \tau_x$ state conversion problem as
    \begin{align*}
        \gammatwo\qty( \langle \xi_x, \xi_y \rangle - \langle \tau_x, \tau_y \rangle \,\bigg|\, \id - O_x^* O_y ) \ .
    \end{align*}
\end{definition}
{Note that also here we used $\langle \xi_x, \xi_y \rangle - \langle \tau_x, \tau_y \rangle$ and $\id - O_x^* O_y$ to actually denote the matrices with these entries.} More precisely, we will take $E = G_{\xi} - G_{\tau}$ (the Gram matrices of $\xi$ and $\tau$), and the blocks of $\Delta$ to be $\Delta_{xy} = \id - O_x^* O_y$. A feasible solution $v = \qty{ v_x }_x$ to this optimization problem is called \emph{catalyst}, {and we will see that this object contains all the information to construct a quantum query algorithm for the given state conversion problem. The filter matrices on the other hand keep the information about the oracle we have access to}.

\begin{theorem}[Proposition 6.10 in~\cite{belovsOneWayTicketVegas2023}]
    The adversary bound always admits a catalyst.
\end{theorem}

Notice that the constraint (\ref{eq:unidir-gamma-bound-constraint-partial}) of the partial bound can be rewritten as
\begin{align*}
    E - G_{\sigma} = V^* (\Delta \otimes \id_{\calW}) V
\end{align*}
for some $G_{\sigma} \succeq 0$, which will be the Gram matrix of the complementary states $\sigma = \qty{ \sigma_x }_{x \in D}$. In other words, the choice of the catalyst $v$ will also influence $\sigma$, and its dimension will be equal to the rank of the residual $E - V^* (\Delta \otimes \id_{\calW}) V$\footnote{We remark that whenever $D$ is finite, $\sigma$ can always be chosen to be one-dimensional, but this does not necessarily mean that the resulting catalyst will be optimal.}.

\begin{lemma}
    \label{thm:total-partial-equivalence}
    If $\norm{\tau_x} = \norm{\xi_x}$ for every $x \in D$, then $\gammatwo = \gammatwoplus$.
\end{lemma}
\begin{proof}
    Let $v$ be a feasible catalyst on the partial $\gamma_2$-bound, and let $G_{\sigma}$ satisfy $$(G_{\sigma})_{xy} = e_{xy} - \langle v_x, \qty( (\id - O_x^* O_y) \otimes \id_{\calW}) v_y \rangle \ .$$
    By~(\ref{eq:unidir-gamma-bound-constraint-partial}) we have $G_{\sigma} \succeq 0$. Since we have that $E$ has a zero diagonal by the hypothesis, while $\id - O_x^* O_x = 0$ by unitarity of the oracle, we conclude that also $G_{\sigma}$ has a zero diagonal, and thus $G_{\sigma} = 0$.
\end{proof}
From now on, when we talk about the adversary bound, we will not make a distinction between total and partial $\gamma_2$-bound and implicitly consider always the latter, since in the case of total state conversion, the feasible solution spaces of the two variants naturally coincide. Another important property states that, if we want a subnormalized $\tau$, the value of the $\gamma_2$-bound decreases accordingly.
\begin{theorem}
    \label{thm:gamma-bound-subnorm-property}
    If $\gammatwo$ is the value of the adversary bound for $\xi \mapsto \tau$ state conversion, then for any $c \in (0, 1)$, the value of the adversary bound for $\xi \mapsto c \tau$ conversion is {at most} $c^2 \gammatwo$.
\end{theorem}
\begin{proof}
    Assuming we have a feasible catalyst $v$ for $\xi \mapsto \tau$ conversion, we have
    \begin{align}
        \langle \xi_x, \xi_y \rangle - \langle \tau_x, \tau_y \rangle \succeq \langle v_x, \left( \Delta_{xy} \otimes \id \right) v_y \rangle \ . \label{eq:gamma-bound-subnorm-property-1}
    \end{align}
    On the other hand, the quantity for $\xi \mapsto c\tau$ can be rewritten as
    \begin{align*}
        \langle \xi_x, \xi_y \rangle - \langle c \tau_x, c \tau_y \rangle & = c^2 \qty(\langle \xi_x, \xi_y \rangle - \langle \tau_x, \tau_y \rangle) + (1 - c^2) \langle \xi_x, \xi_y \rangle
    \end{align*}
    Letting $(G_{\sigma})_{xy} = \langle \sqrt{1 - c^2} \xi_x, \sqrt{1 - c^2} \xi_y \rangle$, and by using (\ref{eq:gamma-bound-subnorm-property-1}) we obtain
    \begin{align*}
        \langle \xi_x, \xi_y \rangle - \langle c \tau_x, c \tau_y \rangle - G_{\sigma} & = c^2 \qty(\langle \xi_x, \xi_y \rangle - \langle \tau_x, \tau_y \rangle) \\
        & \succeq c^2 \langle v_x, \left( \Delta_{xy} \otimes \id \right) v_y \rangle \\
        & = \langle c v_x, \left( \Delta_{xy} \otimes \id \right) c v_y \rangle
    \end{align*}
    concluding that $cv$ is a valid catalyst for $\xi \mapsto c\tau$ conversion.
\end{proof}

\subsection{Lower bounds and upper bounds for query complexity}
\label{sec:query-complexity-bounds}

We now show the well-established result that the adversary bound characterizes the query complexity of state conversion, extending the claim also to the partial case.
\begin{theorem}[Extension of~\cite{leeQuantumQueryComplexity2011,belovsOneWayTicketVegas2023}]
    \label{thm:gamma-monte-carlo-bounds}
    {Consider an instance $\xi \mapsto \tau$ of (partial) state conversion. If $\gammatwo$ is the optimal value of the corresponding adversary bound, the following lower bound holds:
    \begin{align*}
        \gammatwo \le M(\xi \mapsto \tau, O) 
    \end{align*}
    where $M(\xi \mapsto \tau, O)$ is the Monte Carlo complexity for exact state conversion. Also the following upper bound holds:
    \begin{align*}
        M_{\epsilon}(\xi \mapsto \tau, O) \le \bigO(\gammatwo / \epsilon^2)
    \end{align*}
    where $M_{\epsilon}(\xi \mapsto \tau, O)$ is the Monte Carlo complexity of $\epsilon$-approximate state conversion.}
\end{theorem}

\begin{proof}[Proof of lower bound]
    Let $\calA(O) = U_T \Tilde{O} U_{T-1} \cdots \Tilde{O} U_0$ be a valid procedure for $\xi \mapsto \tau \oplus \sigma$ state conversion, and let {$\ket*{\gamma^k_x} := U_k \Tilde{O} \cdots \Tilde{O} U_0 \ket{\xi_x}$} be the intermediate state constructed in $k$ steps. Since $\langle \tau_x \oplus \sigma_x, \tau_y \oplus \sigma_y \rangle = \langle \tau_x,\tau_y \rangle + \langle \sigma_x,\sigma_y \rangle$, we can rewrite the difference as
    \begin{align*}
        \langle \xi_x, \xi_y\rangle - \langle \tau_x,\tau_y \rangle - \langle \sigma_x,\sigma_y \rangle & = \sum_{k = 0}^{T-1} \langle \gamma^{k}_x, \gamma^{k}_y \rangle - \langle \gamma^{k+1}_x, \gamma^{k+1}_y \rangle \\
        & = \sum_{k = 0}^{T-1} \langle\gamma^{k}_x, \gamma^{k}_y \rangle - \langle \Tilde{O}_x \gamma^{k}_x, \Tilde{O}_y \gamma^{k}_y \rangle \\
        & = \sum_{k = 0}^{T-1} \langle \gamma^{k}_x, \qty(\id - \Tilde{O}_x^* \Tilde{O}_y) \gamma^{k}_y \rangle \\
        & = { \sum_{k = 0}^{T-1} \langle \Pi \gamma^{k}_x, \qty[\id - (O_x^* O_y) \otimes \id_\calW] \Pi \gamma^{k}_y \rangle} \ .
    \end{align*}
    {where $\Pi$ is the projector onto $\calM \otimes \calW$. Therefore the direct sum $v_x = \Pi \gamma_x^0 \oplus \cdots \oplus \Pi \gamma_x^{T-1}$ is a valid catalyst satisfying
    \begin{align}
        \label{eq:las-vegas-estimate}
        \norm{v_x}^2 = \sum_{k = 0}^{T-1} \norm{\Pi \gamma^k_x}^2 \le T
    \end{align}}
    while $0 \preceq G_{\sigma} = E - V^* (\Delta \otimes \id_{\calW}) V$ is the positive-semidefinite matrix satisfying the constraint (in the total case we have $G_{\sigma} = 0$, by Theorem~\ref{thm:total-partial-equivalence}).
\end{proof}
We now show a proof of the upper bound, i.e., given a catalyst $v$ for a state conversion problem with value $L = \max_{x \in D} \norm{v_x}^2$, we construct an algorithm with query complexity $\lceil 4L / \epsilon^2 \rceil$.

\begin{proof}[Proof of upper bound]
    By feasibility of $v$ we obtain
    \begin{align*}
        \langle \xi_x, \xi_y \rangle - \langle \tau_x, \tau_y \rangle - \langle \sigma_x, \sigma_y \rangle = \langle v_x, \qty( (\id - O_x^* O_y) \otimes \id_{\calW}) v_y \rangle
    \end{align*}
    for some complementary state $\sigma$, which exists because $G_{\sigma} = E - V^* \Delta V \succeq 0$, as per (\ref{eq:unidir-gamma-bound-constraint-partial}). By rearranging the terms, we can rewrite the constraint as
    \begin{align*}
        \langle \xi_x, \xi_y \rangle + \langle O_x v_x, O_y v_y \rangle = \langle \tau_x, \tau_y \rangle + \langle \sigma_x, \sigma_y \rangle + \langle v_x, v_y \rangle \ .
    \end{align*}
    {Here $O_x$ should be actually $O_x \otimes \id_{\calW}$, but from now on we will keep $\id_{\calW}$ implicit for ease of notation.}
    Since these are respectively the Gram matrices of $\xi \oplus O v$ and $\tau \oplus \sigma \oplus v$, there exists a unitary $S$ satisfying
    \begin{align*}
        S (\xi_x \oplus O_x v_x) = \tau_x \oplus \sigma_x \oplus v_x
    \end{align*}
    for every $x \in D$. From now on we incorporate the complementary state into the target state $\tau \leftarrow \tau \oplus \sigma$ to simplify the notation. We have that, with only one query, we can convert $\xi$ into $\tau$, provided that we have the catalyst. However, the catalyst can have large norm, and we do not really know how to construct it (this would be another state conversion problem). We can nonetheless use $S$ to construct an algorithm $\calA$ which carries out the transformation
    \begin{align*}
        \calA \qty(\xi_x \oplus \frac{v_x}{\sqrt{T}}) = \tau_x \oplus \frac{v_x}{\sqrt{T}}
    \end{align*}
    using $T$ queries. If $T$ is large enough, $v_x/\sqrt{T}$ will be small enough in norm and can be neglected. More precisely, taking {$T = \lceil 4 L/\epsilon^2 \rceil$}, where $L = \max_{x \in D} \norm{ v_x }^2$, we have an error bound
    \begin{align*}
        \norm{\calA \xi_x - \tau_x} \le \norm{\xi_x - \xi_x \oplus \frac{v_x}{\sqrt{T}}} + \norm{\tau_x - \tau_x \oplus \frac{v_x}{\sqrt{T}}} \le \epsilon \ .
    \end{align*}
    Since $\calA$ is unitary, the error is preserved when we go from $\xi \oplus v / \sqrt{T}$ to $\tau \oplus v / \sqrt{T}$.

    We only need to construct $\calA$: what we do is to divide the initial state
    \begin{align*}
        \xi_x \oplus \frac{v_x}{\sqrt{T}} \mapsto \underbrace{\frac{\xi_x}{\sqrt{T}} \oplus \cdots \oplus \frac{\xi_x}{\sqrt{T}}}_{T \text{ times}} \oplus \frac{v_x}{\sqrt{T}}
    \end{align*}
    which in terms of quantum states requires the transformation $$\ket{1}_A \ket{\xi_x}_B + \frac{1}{\sqrt{T}} \ket{0}_A \ket{v_x}_B \mapsto \frac{1}{\sqrt{T}} \sum_{t = 1}^T \ket{t}_A \ket{\xi_x}_B + \frac{1}{\sqrt{T}} \ket{0}_A \ket{v_x}_B \ .$$ We then apply $S (\id \oplus O)$ to each of these pieces, together with the unique copy of $v_x /\sqrt{T}$ (in other words, at the $k$-th step the $S (\id \oplus O)$ will be controlled on $\ket{k}_A, \ket{0}_A$). Therefore, after $k$ steps the state will be of the form
    \begin{align*}
        \underbrace{\frac{\tau_x}{\sqrt{T}} \oplus \cdots \oplus \frac{\tau_x}{\sqrt{T}}}_{k \text{ times}} \oplus \underbrace{\frac{\xi_x}{\sqrt{T}} \oplus \cdots \oplus \frac{\xi_x}{\sqrt{T}}}_{T - k \text{ times}} \oplus \frac{v_x}{\sqrt{T}}
    \end{align*}
    and after $T$ steps we can recompose the copies of $\tau_x$ (i.e., undo the superposition on register $A$). While the Monte Carlo complexity of this procedure is exactly $T$, the total Las Vegas complexity is $T \cdot \norm{v_x}^2/T = \norm{v_x}^2$.
\end{proof}
\begin{algorithm}
    \caption{State conversion algorithm~\cite{belovsOneWayTicketVegas2023}}
    \label{alg:state-conversion}
    \begin{algorithmic}
        \Require The starting state $\ket{\xi_x}_B$ and access to an oracle $O_x$.
        \Ensure The target state $\ket{\tau_x}_B$, up to $\epsilon$ error.
        \State Let $T = \lceil 4L/\epsilon^2 \rceil$.
        \State Initial state is $\ket{1}_A \ket{\xi_x}_B$.
        \Comment $B$ is sufficiently large to contain $v_x$.
        \State Prepare the superposition {$\ket{1}_A \mapsto \frac{1}{\sqrt{T}} \sum_{k = 1}^T \ket{k}_A$}.

        \For{$k = 1$ \textbf{to} $T$}
            \State Apply $O_x$ controlled on $\ket{0}_A$.
            \State Apply $S$ to $\vspan{ \ket{k}_A, \ket{0}_A }$.
            \Comment $S(\id \oplus O_x) (\xi_x \oplus v_x) = \tau_x \oplus v_x$
        \EndFor
        \State Undo the superposition $\frac{1}{\sqrt{T}} \sum_{k = 1}^T \ket{k}_A \mapsto \ket{0}_A$.
    \end{algorithmic}
\end{algorithm}
{The intuition behind the outlined procedure (summarized in Algorithm~\ref{alg:state-conversion}) is similar to the one behind Theorem~\ref{thm:gamma-bound-subnorm-property}: instead of doing $\xi \oplus v \mapsto \tau \oplus v$ state conversion, we carry out $N$ different partial state conversions $\xi \oplus v \mapsto \frac{1}{\sqrt{N}} (\tau \oplus v)$ each of which has a cost $\gammatwo/N$.}
While the claim of Theorem~\ref{thm:gamma-monte-carlo-bounds} bounds the Monte Carlo query complexity for $\epsilon$-approximate state conversion, the result in~\cite{belovsOneWayTicketVegas2023} is stronger when it comes to Las Vegas complexity.
\begin{theorem}[\cite{belovsOneWayTicketVegas2023}]
    \label{thm:gamma-las-vegas-bounds}
    If $L(\xi \mapsto \tau, O)$ is the Las Vegas complexity of $\xi \rightarrow \tau$ exact state conversion with oracle $O$, and $\gammatwo$ is the corresponding value of the adversary bound, then $L(\xi \mapsto \tau, O) = \gammatwo$.
\end{theorem}
The proof of the lower bound is exactly as the one for Theorem~\ref{thm:gamma-monte-carlo-bounds}, with the only difference that {the quantity in (\ref{eq:las-vegas-estimate}) is exactly the Las Vegas complexity of the given algorithm}. For the upper bound, the algorithm basically involves the chain of transformations
\begin{align*}
    \xi_x \mapsto \xi_x \oplus \frac{v_x}{\sqrt{T}} \mapsto \tau_x \oplus \frac{v_x}{\sqrt{T}} \mapsto \tau_x
\end{align*}
where the middle transformation is given by Algorithm~\ref{alg:state-conversion} (the construction is exact since we now have $v_x$), while the first and third transformations require a Las Vegas complexity that goes to $0$ as $T \rightarrow \infty$ (the Monte Carlo complexity, however, grows with $T$). We do not really need this Las Vegas construction for our exposition, but as mentioned earlier, Las Vegas complexity has nice composition properties that might be useful even in our case, and we invite the interested reader to check the full argument in~\cite{belovsOneWayTicketVegas2023}.

\section{Quantum signal processing}
\label{sec:qsp}

In this section we review quantum signal processing (QSP).
{We start with univariate QSP over $SU(2)$ and introduce the layer stripping argument in~\ref{sec:qsp/univariate-qsp-layer-stripping}. We then extend the formalism from $SU(2)$ to $SU(r+1)$ with $r > 1$, and show in~\ref{sec:qsp-over-sun} that essentially the same results hold as long as we keep one variable~\cite{laneveQuantumSignalProcessing2024,luQuantumSignalProcessing2024}. We conclude this section by introducing multivariate QSP and its caveats in~\ref{sec:mqsp-intro}.}

In particular, we start with univariate quantum signal processing over $SU(2)$, and then extend this over arbitrary dimensions and multiple variables~\cite{rossiMultivariableQuantumSignal2022,moriCommentMultivariableQuantum2024,nemethVariantsMultivariateQuantum2023,
laneveMultivariatePolynomialsAchievable2025,gomesMultivariableQSPBosonic2024}, which will be relevant for the rest of this work.

\subsection{Univariate QSP and the layer stripping argument}
\label{sec:qsp/univariate-qsp-layer-stripping}
{In quantum signal processing over $SU(2)$ we have a complex number $z \in \T$ (the signal), embedded in a generic single-qubit unitary of the form
\begin{align*}
    \Tilde{w} = \begin{bmatrix}
        z & \\
        & 1
    \end{bmatrix}
\end{align*}
which we call \emph{signal operator}. We then choose a sequence $A = (A_0, A_1, \ldots, A_n)$ of $n+1$ single-qubit unitaries $A_k \in SU(2)$ (the \emph{processing operators}) and carry out the matrix multiplication
\begin{align*}
    \calP_A(z) := A_n \Tilde{w} A_{n-1} \Tilde{w} \cdots \Tilde{w} A_1 \Tilde{w} A_0
\end{align*}
We say that $\calP_A(z)$ is a \emph{QSP protocol} of length $n$. One can see that if we start from some fixed quantum state, e.g., $\ket{0}$, then
\begin{align*}
    \calP_A(z) \ket{0} = P(z) \ket{0} + Q(z) \ket{1}
\end{align*}
where the amplitudes $P(z), Q(z)$ are polynomials in $z$ of degree $\le n$ satisfying the normalization condition $|P(z)|^2 + |Q(z)|^2 = 1$ for $z \in \T$, by construction. The polynomials are determined by the $A_k$'s we choose.}

A natural question arises: given normalized $(P(z), Q(z))$ of degree $n$, can we \emph{always} construct such pair with a QSP protocol of length $\le n$? The answer to this question was proven to be positive.
\begin{theorem}[\cite{motlaghGeneralizedQuantumSignal2024}]
    \label{thm:gqsp-theorem}
    For any $P(z), Q(z) \in \C[z]$ of degree $\le n$ with $|P(z)|^2 + |Q(z)|^2 = 1$ on $\T$, a sequence $A = (A_0, A_1, \ldots, A_n)$ with $A_k \in SU(2)$ exists such that
    \begin{align*}
        \calP_A(z) \ket{0} = P(z) \ket{0} + Q(z) \ket{1} \ .
    \end{align*}
\end{theorem}
\noindent The proof of this result is usually by induction, and is known as the \emph{layer stripping argument}.
\begin{proof}
    Let $\tau(z) = P(z) \ket{0} + Q(z) \ket{1}$ and write it down as a polynomial state
    \begin{align*}
        \tau(z) = \sum_{k = 0}^n \tau_k z^k \ .
    \end{align*}
    For $n = 0$ the claim is trivial, as $\tau(z)$ would be a constant quantum state, which can be written as $A_0 \ket{0}$, by taking $\tau$ to be the first column of $A_0$.

    For $n > 0$ the normalization condition $\langle \tau(z), \tau(z) \rangle = 1$ must hold as a polynomial equation, {whose $n$-th degree term} yields
    \begin{align}
        \langle \tau_0, \tau_n \rangle = 0 \ , \label{eq:gqsp-coeff-orthogonality}
    \end{align}
    i.e., the {two (vector) coefficients} of $\tau(z)$ at the two endpoints must be orthogonal. We can thus choose a unitary $A_n$ such that
    \begin{align*}
        A_n^\dag \tau_n & \propto \ket{0} \\
        A_n^\dag \tau_0 & \propto \ket{1}
    \end{align*}
    and, {since phases we put on $\ket{0}, \ket{1}$ don't matter, we can choose $A_n$ to have determinant $1$ so that $A_n \in SU(2)$}. This implies that with $A_n^\dag \tau(z) := P'(z) \ket{0} + Q'(z) \ket{1}$ we have that $P'$ has only degrees in $\{ 1, \ldots, n \}$, while $Q'$ has only degrees in $\{ 0, \ldots, n-1 \}$. Therefore applying $\Tilde{w}^\dag$ will realign the coefficients in the same range $\{ 0, \ldots, n-1 \}$. Thus we stripped a layer, obtaining a polynomial vector $\tau'(z) := \Tilde{w}^\dag A_n^\dag \tau(z)$ of degree $n-1$, concluding the proof with an application of the induction hypothesis.
\end{proof}
The one presented is a specific variant of QSP, called \emph{generalized QSP} (or GQSP)~\cite{motlaghGeneralizedQuantumSignal2024}. Sometimes, we call this \emph{analytic QSP} by the shape of the signal operator, which results in $P(z), Q(z)$ being analytic (i.e., not Laurent) polynomials. This variant of QSP is perfectly equivalent to other variants of QSP found in the literature~\cite{martynGrandUnificationQuantum2021,haahProductDecompositionPeriodic2019} (see~\cite{laneveGeneralizedQuantumSignal2025} for a more comprehensive treatment of these equivalences), which in turn can be lifted to the celebrated \emph{quantum singular value transformation} (or QSVT)~\cite{gilyenQuantumSingularValue2019}. Also GQSP can be directly lifted to an improved version of the QSVT, called \emph{generalized QSVT}~\cite{sunderhaufGeneralizedQuantumSingular2023}, which lead to more expressiveness for the polynomial transformations of the singular values of the block-encoded matrices.

\subsection{Fej{\'e}r-Riesz factorization and the completion problem}
Usually we are interested in one polynomial (say $P$), which would correspond to the desired transformation we want to carry out on the eigenvalues or singular values of given block-encoded matrices. {We then have to find a \emph{complementary polynomial} $Q$ such that $|P|^2 + |Q|^2 = 1$ on $\T$ so that we can apply Theorem~\ref{thm:gqsp-theorem}: the following theorem guarantees the existence of a single polynomial.}
\begin{theorem}[Fej{\'e}r-Riesz~\cite{hussenFejerRieszTheoremIts2021}]
    Let $P(z)$ be a $n$-degree polynomial satisfying $|P(z)| \le 1$ on $\T$. There exists a $n$-degree complementary $Q(z)$ such that $|P(z)|^2 + |Q(z)|^2 = 1$.
\end{theorem}
Such complementary $Q$ is unique up to a phase (which we neglect), and a transformation $\alpha \mapsto 1/\alpha^*$ of {a subset $S$ (containing possible repetitions) of its roots}, which preserves $|Q|^2$. Thus, a given $n$-degree $P$ has $2^n$ possible complementary polynomials\footnote{If some root appears multiple times, then some of these $2^n$ polynomials are actually equal.}
\begin{align*}
    {Q_S(z) = Q_{\emptyset}(z) \prod_{\alpha \in S} \frac{\alpha}{|\alpha|} \frac{\alpha^* z - 1}{z - \alpha}}
\end{align*}
where $Q_{\emptyset}$ is the unique \emph{outer} complementary polynomial, i.e., the one with all the roots outside $\T$, and the product is the \emph{Blaschke product} of the roots {of $Q_\emptyset$} we want to move from outside to inside $\T$ (here $S$ is the multiset of roots we want to move). {Note that $Q_S(z)$ is always a polynomial: if $\alpha \in S$, then $\alpha$ is a root of $Q_{\emptyset}$, and thus $(z - \alpha)$ is one of its factors. The Blaschke product replaces this factor with $(1 - \alpha^* z)$, effectively doing the $\alpha \mapsto 1/\alpha^*$ transformation.} Moreover, the Blaschke product has unit modulus on $\T$, which guarantees $|Q_S| = |Q_{\emptyset}|$ on $\T$, i.e., $Q_S$ is also a valid complementary polynomial.

{\begin{corollary}
    \label{thm:gqsp-completeness}
    Given any $P(z) \in \C[z]$ of degree $\le n$ satisfying $|P(z)| \le 1$ on $\T$, there exists $A = (A_0, \ldots, A_n)$ with $A_k \in SU(2)$ such that
    \begin{align*}
        \bra{0} \calP_A(z) \ket{0} = P(z) \ .
    \end{align*}
\end{corollary}}

Methods to numerically estimate a complementary $Q$ from $P$ are known~\cite{haahProductDecompositionPeriodic2019,motlaghGeneralizedQuantumSignal2024,berntsonComplementaryPolynomialsQuantum2025,alexisInfiniteQuantumSignal2024}, and constitute a crucial part in the numerical computation of a QSP protocol for a desired $P$. The fact that any $P$ admits a single complementary polynomial allows any polynomial to be embedded into $SU(2)$ matrices, and thus decomposable into a $SU(2)$-QSP protocol.

\subsection{Extension to more than two dimensions}
\label{sec:qsp-over-sun}

{The above ansatz can be generalized from $SU(2)$ to $SU(r+1)$ with $r > 1$. We have a signal operator
\begin{align}
    \Tilde{w}^{(r)} := \diag(z, \underbrace{1, \ldots, 1}_r) \label{eq:qsp-sun-signal-operator}
\end{align}
and we intertwine this with a sequence $A = (A_0, \ldots, A_n)$ with $A_k \in SU(r+1)$.
\begin{align*}
    \calP_A(z) = A_n \Tilde{w}^{(r)} \cdots \Tilde{w}^{(r)} A_0 \ .
\end{align*}
By a similar layer stripping argument, one can obtain a result analogous to Theorem~\ref{thm:gqsp-theorem}.
\begin{theorem}[\cite{laneveQuantumSignalProcessing2024}]
    \label{thm:qsp-sun-theorem}
    Any polynomial state $\tau(z) = P_k(z) \ket{0} \sum_{k=1}^r Q_k(z) \ket{k}$ of degree $\le n$ admits a sequence $A = (A_0, \ldots, A_n)$ with $A_k \in SU(N)$ such that
    $$\calP_A(z) \ket{0} = \tau(z) \ .$$
\end{theorem}}
This statement has the same proof as Theorem~\ref{thm:gqsp-theorem}, with the slight difference that we have a larger space of valid $A_n$ we can choose from. Besides that, it is sufficient to show that $\tau_0$ gets multiplied by $1$ and $\tau_n$ {gets} multiplied by $z^{-1}$ upon application of $(\Tilde{w}^{(r)})^\dag A_n^\dag$ in order to obtain a degree reduction. Thus, the actual number of ones and $z$'s in~(\ref{eq:qsp-sun-signal-operator}) does not really change the expressiveness of the ansatz, as long as there is at least one occurrence of each along the diagonal.

The fact that Corollary~\ref{thm:gqsp-completeness} ensures that any polynomial can be decomposed even on a single qubit limits the usefulness of $SU(r+1)$-QSP in the univariate case, we will see that considering more than two dimensions can be of help when introducing multiple variables.

\subsection{Multivariate QSP}
\label{sec:mqsp-intro}
In a multivariate QSP (M-QSP) setting we have a vector of $m > 1$ signals $\vv{z} := (z_1, \ldots, z_m) \in \T^m$ we would like to polynomially transform. There are different proposals across the literature which differ mainly on the shape of the signal operator. {The first one, proposed in~\cite{rossiMultivariableQuantumSignal2022}, is a protocol over $SU(2)$ where we define $m$ different signal operators $\Tilde{w}_j = \diag(z_j, 1)$\footnote{Our exposition is slightly different than the one in~\cite{rossiMultivariableQuantumSignal2022}, since we are using the \emph{analytic picture} $\diag(z, 1)$, as opposed to their \emph{Laurent} formulation $\diag(z, z^{-1})$. The two ans{\"a}tze, however, are equivalent, see~\cite{laneveGeneralizedQuantumSignal2025}.}. Given a sequence $A = (A_0, \ldots, A_n)$ with $A_k \in SU(2)$, and a string $s = (s_1, \ldots, s_n)$ with $s \in \{ 1, \ldots, m \}$, the M-QSP protocol is defined as follows:
\begin{align}
    \calP_{A, s}(\vv{z}) = A_n \Tilde{w}_{s_n} A_{n-1} \Tilde{w}_{s_{n-1}} \cdots \Tilde{w}_{s_1} A_0 \label{eq:rossi-chuang-qsp}
\end{align}
i.e., at each step we are allowed to choose one variable $z_j$ to call, and the result will be a pair of $m$-variate polynomials.
\begin{align*}
    \calP_{A,s}(\vv{z}) \ket{0} = P(\vv{z}) \ket{0} + Q(\vv{z}) \ket{1}
\end{align*}
}Other works~\cite{nemethVariantsMultivariateQuantum2023,laneveMultivariatePolynomialsAchievable2025} define protocols over $SU(m+1)$ with a single, fixed signal operator, e.g.,
\begin{align*}
    \Tilde{W} = \diag(z_1, z_2, \ldots, z_m, 1) \ .
\end{align*}
The main problem with M-QSP is that we do not have an analogue of Theorem~\ref{thm:gqsp-theorem}. The layer stripping argument does not apply, as the counterpart of~(\ref{eq:gqsp-coeff-orthogonality}) that would let us strip a layer does not follow directly from the normalization condition, and thus we would need additional conditions that guarantee the success of the induction step.

As a consequence, attempts to define an ansatz for QSP over multiple variables retaining the single-qubit structure fail to give the strong guarantees {of Corollary~\ref{thm:gqsp-completeness}}, namely that any polynomial $P(\vv{z})$ satisfying $|P(\vv{z})| \le 1$ can be constructed: it was proven in~\cite{nemethVariantsMultivariateQuantum2023} that there are some normalized pairs $(P, Q)$ that are not decomposable {into a M-QSP protocol over $SU(2)$}.

Even a multivariate counterpart of the Fej{\'e}r-Riesz theorem has its own caveats: necessary and sufficient conditions for a given bivariate $P(\vv{z})$ to have a single complementary $Q$ are known, but are quite restrictive~\cite{geronimoPositiveExtensionsFejerRiesz2004,hussenFejerRieszTheoremIts2021,rossiMultivariableQuantumSignal2022}. The reason mainly stems from the fact that the univariate Fej{\'e}r-Riesz theorem heavily relies on the fundamental theorem of algebra. However, we can obtain a reasonable version of the statement if we allow more than two dimensions.
\begin{theorem}[Multivariate Fej{\'e}r-Riesz, adapted from~\cite{dritschelFactorizationTrigonometricPolynomials2004}]
    \label{thm:fejer-riesz-multivariate}
    Let $P(\vv{z})$ be a $m$-variate polynomial satisfying $|P(\vv{z})| < 1$ on $\T^m$. Then {for some $r > 0$} there exist $Q_1, \ldots, Q_r$ such that
    \begin{align*}
        |P(\vv{z})|^2 + \sum_{k = 1}^r |Q_k(\vv{z})|^2 = 1 \text{\ \ \ for every $\vv{z} \in \T^m$.}
    \end{align*}
\end{theorem}
{We remark that we do not have a clear upper bound on $r$ in terms of $m$ or the degrees of the polynomials, we only know that it is finite.} Also notice that we require $|P(\vv{z})| < 1$ strictly, which is a weaker statement than the univariate counterpart. It was shown in~\cite{rudinExtensionProblemPositivedefinite1963} that there are polynomials with unit supremum norm on $\T^m$ that do not admit a Fej{\'e}r-Riesz factorization at all.

If we want some $P$, however, it might still be possible to find a full M-QSP decomposition using different complementary polynomials, perhaps over more dimensions. Theorem~\ref{thm:fejer-riesz-multivariate} thus motivates the exploration of M-QSP protocols over more than two dimensions and, for the purpose of this work, we define a more general ansatz for M-QSP over $SU(r+1)$: {we choose a sequence $A = (A_0, A_1, \ldots, A_n)$ with $A_k \in SU(r+1)$ as usual, and for the diagonal of the signal operator we choose $r$ symbols (possibly with repetition) from $\{ z_1, \ldots, z_m \}$. We define $\Tilde{w}^{(r)}_S$ to be the diagonal containing the elements of the multiset $S \subseteq \{ z_1, \ldots, z_m \}$ (in some fixed order, possibly padded with ones), and we define
\begin{align*}
    \calP_{A, S}(\vv{z}) = A_n \Tilde{w}^{(r)}_{S_n} A_{n-1} \Tilde{w}^{(r)}_{S_{n-1}} \cdots \Tilde{w}^{(r)}_{S_0} A_0
\end{align*}
for a set $S = (S_1, \ldots, S_n)$ of multisets.} An example of choices in a $SU(3)$ protocol with $m = 2$ variables would be
\begin{align*}
    A_3
    \begin{bmatrix}
        z_1 & & \\
        & z_1 & \\
        & & z_2
    \end{bmatrix}
    A_2
    \begin{bmatrix}
        z_2 & & \\
        & 1 & \\
        & & 1
    \end{bmatrix}
    A_1
    \begin{bmatrix}
        z_2 & & \\
        & z_2 & \\
        & & 1
    \end{bmatrix}
    A_0
\end{align*}
{by taking $S_1 = \{ z_2, z_2 \}, S_2 = \{ z_2 \}, S_3 = \{ z_1, z_1, z_2 \}$.}
The order of the symbols along the diagonal does not matter, since the $A_k$ can always be tweaked to swap the elements of the diagonals. Clearly, all the aforementioned constructions are special cases of this ansatz.

\section{Quantum signal processing as a state conversion problem}
\label{sec:state-conversion-l2}

{This section will merge the theories reviewed in the previous two sections. We start by introducing the $L^2$ space and linear operators on it (Appendix~\ref{apx:trace-class} gives a quick introduction to the elements of operator theory we need). \ref{sec:gamma-bounds-l2} then proceeds to redefine the $\gamma_2$-bound and the adversary bound on this space. We then show in~\ref{sec:poly-state-conversion} further properties of the feasible space of the adversary bound when state conversion is between polynomials. Then~\ref{sec:qsp-as-state-conversion} will connect QSP and show one of the main technical results of the paper: a one-to-one correspondence between the feasible solutions of the adversary bound and the $SU(2)$-QSP protocols. \ref{sec:qsp-sun-as-state-conversion} extends this bijection to protocols over more than two dimensions, and explains the role of partial state conversion in also finding complementary polynomial(s) for QSP. We conclude with~\ref{sec:mqsp-as-state-conversion} by extending this formalism to M-QSP, discussing the implications.}

State conversion was always studied across the literature over some finite set of labels $D$~\cite{leeQuantumQueryComplexity2011,belovsOneWayTicketVegas2023}. As we reported in Section~\ref{sec:finite-state-conversion}, the theory of quantum query complexity led to the definition of the $\gamma_2$-bounds, which are convex optimizations over {linear operators acting on} the Hilbert space $\C^{|D|}$. One can equivalently think of $\xi, \tau$ as functions $D \rightarrow \C^K$, for some fixed dimension $K$, and this allows us to take one step further: {QSP as presented in the previous section looks like a quantum query algorithm, with the only exception that the oracle $\Tilde{w}$ is determined by the signal $z$ or $\vv{z}$, which is a continuous variable}.

Thus, we can consider the set of labels $D = \T^m$, i.e., a label $\vv{z} \in D$ is a tuple $(z_1, \ldots, z_m)$ where $|z_k| = 1$ for every $k$. We then consider $\xi, \tau$ to be vectors of functions in $L^{\infty}(\T^m)$.
\begin{align*}
    \xi, \tau : \T^m & \mapsto \C^K \\
    \vv{z} & \mapsto \xi(\vv{z}), \tau(\vv{z})
\end{align*}
For the rest of the work we will use the notation $\xi(\vv{z})$ instead of $\xi_{\vv{z}}$ to highlight that we are dealing with functions of continuous variable. We highlight that, although we now have an infinite set of labels $D$, the states are still in $\C^K$, a finite-dimensional space. {We will need to dive into the theory of linear operators over $L^2(\T^m)$, the Hilbert space of square-integrable functions on $\T^m$, for which the unfamiliar reader can find a quick and self-contained introduction in Appendix~\ref{apx:trace-class}.}

We consider the usual inner product on $L^2(\T^m)$ with respect to the normalized Lebesgue measure $\mu$ on $\T^m$\footnote{$L^\infty(\T^m) \subseteq L^2(\T^m)$ since $\T^m$ has finite measure {(note that the opposite inclusion is not true)}.}.
\begin{align*}
    \langle f, g \rangle =&\ \int_{\T^m} \bar{f} g \ \dd\mu \\
    := &\ \frac{1}{(2\pi)^m} \int_0^{2\pi} \cdots \int_0^{2\pi} \overline{f(e^{i\theta_1}, \ldots, e^{\theta_m})} g(e^{i\theta_1}, \ldots, e^{\theta_m}) \dd{\theta_1} \cdots \dd{\theta_m} \ .
\end{align*}
{A complete orthonormal basis for $L^2(\T^m)$ is the Fourier basis $\{ t_{\vv{k}} \}_{\vv{k} \in \Z^m}$ with:
\begin{align*}
    t_{\vv{k}}(\vv{z}) = z_1^{k_1} z_2^{k_2} \cdots z_m^{k_m}
\end{align*}
of which we will make extensive use.} If $f$ is a vector of $K$ square-integrable functions, then it can be seen as a point in the space $$L^2(\T^m)^K \simeq L^2(\T^m) \otimes \C^K =: L^2(\T^m, \C^K)$$ which is a Hilbert space\footnote{{The isomorphism holds since $(f_1, \ldots, f_K) = \sum_{j = 0}^{K-1} f_j \otimes \ket{j}$, true for any separable Hilbert space $\calH$.}}. More generally, we denote with $L^2(\T^m, \calH)$ the space of square-integrable functions with values in the Hilbert space $\calH$.

Given a linear operator $A$ on $L^2(\T^m)$, its \emph{integral kernel} $a(\vv{x}, \vv{y})$ is a distribution\footnote{By considering integral kernels in the distributional sense, every linear operator can be seen as an integral operator, as in standard quantum mechanics. {A distribution that is not a function is, e.g., the Dirac delta $\delta(x)$.}} satisfying
\begin{align*}
    \langle f, A g \rangle = \iint_{(\T^m)^2} \overline{f(\vv{x})} a(\vv{x}, \vv{y}) g(\vv{y}) \ \dd\mu(\vv{x}) \ \dd\mu(\vv{y})
\end{align*}
for every $f, g \in L^2(\T^m)$. This also applies to block operators, where $A$ is an operator on $L^2(\T^m, \C^K)$ and $a(\vv{x}, \vv{y})$ is a matrix on $\C^K$.

\subsection{The $\gamma_2$-bounds over $L^2(\T^m, \C^K)$}
\label{sec:gamma-bounds-l2}

We are interested in the conversion $\xi(\vv{z}) \mapsto \tau(\vv{z})$, given access to an oracle $O(\vv{z}) : \T^m \rightarrow U(\calM)$. All the definitions for state conversion over finite labels given in the Section~\ref{sec:finite-state-conversion} are also valid in this setting. Our goal now is to define the $\gamma_2$-bounds over $L^2(\T^m, \C^K)$, preserving as many properties as possible.

\begin{definition}[Unidirectional relative $\gamma_2$-bound over $L^\infty(\T^m)$]
    \label{def:unidir-gamma-bound-l2}
    Fix a linear operator $E$ on $L^2(\T^m)$ with $e(\vv{x}, \vv{y})$ being its integral kernel and a Hermitian bounded block operator $\Delta$ with continuous kernel $\Delta(\vv{x}, \vv{y}): \calM \rightarrow \calM$ for some finite-dimensional Hilbert space $\calM$. Define $\gammatwo(E \,|\, \Delta)$ as the optimal value of the following optimization problem:
    \begin{align}
        \inf \ \ \ \ & \sup_{\vv{z} \in \T^m} \ \lVert v(\vv{z}) \rVert^2\nonumber \\
        \text{subject to} \ \ \ \ & e(\vv{x}, \vv{y}) \succeq \langle v(\vv{x}), (\Delta(\vv{x}, \vv{y}) \otimes \id_{\calW}) v(\vv{y}) \rangle & \text{ for every $\vv{x}, \vv{y} \in \T^m$} \label{eq:unidir-gamma-bound-constraint-partial-l2} \\
        & \text{$\calW$ is a Hilbert space, $v(\vv{z}) \in L^2(\T^m, \calM \otimes \calW)$}\nonumber
    \end{align}
\end{definition}
\noindent This definition is the straightforward generalization of Definition~\ref{def:unidir-gamma-bound-partial}: {similar to the finite case, we use the semidefinite constraint as in~(\ref{eq:unidir-gamma-bound-constraint-partial-l2}) to implicitly mean that $\succeq$ is on the operators with those functions as kernels.}
\begin{definition}
    \label{def:adversary-bound-l2}
    The \emph{adversary bound} for $\xi(\vv{z}) \mapsto \tau(\vv{z})$ state conversion is defined as
    \begin{align*}
        \gammatwo\qty( \langle \xi(\vv{x}), \xi(\vv{y}) \rangle - \langle \tau(\vv{x}), \tau(\vv{y}) \rangle \,\bigg|\, \id - O(\vv{x})^* O(\vv{y}) ) \ .
    \end{align*}
\end{definition}
There are some subtleties, however. First of all, the objective function takes the supremum instead of the maximum (it should take the essential supremum, but we can assume they coincide, as otherwise the catalyst would be sub-optimal). Moreover, this is not a semidefinite program anymore in the traditional sense, and thus some results from classical optimization theory might not hold. In particular, $\calW$ is not even guaranteed to be finite-dimensional, and this might be a problem when we want to embed $v(\vv{z})$ in finitely many qubits. More specifically, the Gram operators $G_{\xi}, G_{\tau}$ have finite rank (since they are Gram operators of functions in $\C^K$). On the other hand, the Gram operators $G_{v}, G_{Ov}$ can be proven to be trace-class (see Appendix~\ref{apx:trace-class}).

Since we are working with trace-class operators, Theorems~\ref{thm:relative-gamma-bound-weak-duality} and~\ref{thm:gamma-bound-subnorm-property}, as well as the lower bound of Theorem~\ref{thm:gamma-monte-carlo-bounds} carry over to this case. {A more formal discussion of these observations, along with proofs, is made in Appendix~\ref{apx:trace-class}.}

In other words, while it is always true that $G_v - G_{Ov}$ has finite rank, this does not imply that $G_v$ does ($G_{Ov}$, however, has finite rank if and only if $G_v$ does). The rank of $G_v$ will give the dimension of $v$, so we are interested in which cases we have finite-rank solutions.

\subsection{Polynomial state conversion}
\label{sec:poly-state-conversion}
We now restrict ourselves to the case where $\xi(\vv{z}), \tau(\vv{z})$ are finite polynomials
\begin{align*}
    \xi(\vv{z}) = \sum_{\vv{k}} \xi_{\vv{k}} \vv{z}^{\vv{k}},\ \ \ \ \tau(\vv{z}) = \sum_{\vv{k}} \tau_{\vv{k}} \vv{z}^{\vv{k}}
\end{align*}
where $\vv{z}^{\vv{k}}$ is a shorthand for $z_1^{k_1} \cdots z_m^{k_m}$. The polynomial case is a very special case, because $\xi$ and $\tau$ are supported only on finitely many elements of the Fourier basis {(i.e., only finitely many $\xi_{\vv{k}}, \tau_{\vv{k}}$ are non-zero)}, and thus we might be able to rewrite the constraint of the $\gamma_2$-bound as a finite system using the coefficients $\xi_{\vv{k}}, \tau_{\vv{k}}$.

We start with the simplest example: a single variable $z$ accessed through the oracle $O(z) = z$. In this simple case, we actually get a feasible space that is well-behaved. All the proofs can be found in Appendix~\ref{apx:catalyst-proofs-univariate}.
\begin{lemma}
    \label{thm:feasible-space-polynomial-univariate}
    If $v(z)$ is a feasible catalyst for a $\xi(z) \mapsto \tau(z)$ polynomial state conversion, where both $\xi, \tau$ have degree bounded by $n$, then $v(z)$ is a polynomial of degree at most $n - 1$.
\end{lemma}
\begin{proof}[Proof Sketch]
    The constraint of~(\ref{eq:unidir-gamma-bound-constraint-partial-l2}) can be rewritten in the Fourier basis as
    \begin{align}
        \langle \xi_k, \xi_h \rangle - \langle \tau_k, \tau_h \rangle = \langle v_k, v_h \rangle - \langle v_{k-1}, v_{h-1} \rangle \label{eq:unidir-gamma-bound-constraint-partial-l2-fourier}
    \end{align}
    for every $k, h \in \Z$. {When $k$ or $h$ are not in $\{ 0, \ldots, n \}$ the left-hand side is zero (since $\xi, \tau$ are supported only on this range)}, which implies, for $k = h$
    $$\langle v_{k}, v_{k} \rangle = \langle v_{k-1}, v_{k-1} \rangle$$
    If $v_n \neq 0$, then $\norm{v_k} = \norm{v_n} \not\rightarrow 0$ for every $k > n$, which contradicts $v \in L^2$.
\end{proof}
By Lemma~\ref{thm:feasible-space-polynomial-univariate}, for a univariate polynomial state conversion $\xi \mapsto \tau$ with degrees $\le n$ and oracle $O(z) = z$, the feasible solution space includes only polynomials of degree $n - 1$. This allows to conclude that the rank of these catalysts is also finite.
\begin{corollary}
    \label{thm:univariate-catalyst-embedding}
    If $v(z)$ is a feasible catalyst for a $\xi(z) \mapsto \tau(z)$ polynomial state conversion, where both $\xi, \tau$ have degree bounded by $n$, then $v(z)$ can be embedded in $\C^n$.
\end{corollary}
\begin{proof}
    By Lemma~\ref{thm:feasible-space-polynomial-univariate}, $v(z)$ is a polynomial of degree $n-1$
    \begin{align*}
        v(z) = \sum_{k = 0}^{n-1} v_k z^k \ .
    \end{align*}
    with $v_k \in \calM \otimes \calW$ ($\calM$ has dimension $1$ in this case). Map $v_0, \ldots, v_{n-1}$ to $u_0, \ldots, u_{n-1} \in \C^n$, in such a way that inner products are preserved. This can be done with, e.g., a Gram-Schmidt orthogonalization, and such isometry will preserve (\ref{eq:unidir-gamma-bound-constraint-partial-l2-fourier}).
\end{proof}

\subsection{Quantum signal processing as a polynomial state conversion}
\label{sec:qsp-as-state-conversion}
The reader might have noticed that the QSP ans{\"a}tze defined in Section~\ref{sec:qsp} can be seen as quantum query algorithms. In particular, the signal operator $\Tilde{w} = \diag(z, 1)$ used in $SU(2)$-QSP is simply the oracle $O(z) = z$ applied to half of the Hilbert space.

At this point, QSP as introduced in Section~\ref{sec:qsp} coincides with the polynomial state conversion $\ket{0} \mapsto \tau(z) = P(z) \ket{0} + Q(z) \ket{1}$. The only difference is that QSP enforces a constraint on the dimensionality of the protocol, as opposed to polynomial state conversion, but we'll see that this does not make any difference, i.e., any protocol for polynomial state conversion will be embeddable in $SU(r)$, where $r$ is the dimension of $\tau(z)$. We start with the $SU(2)$ case.

{
    Lemma~\ref{thm:feasible-space-polynomial-univariate} and Corollary~\ref{thm:univariate-catalyst-embedding} guarantee that any catalyst $v$ will have degree $n-1$ and dimension $n$ if our goal is to build $n$-degree polynomial states. Moreover, similarly to what we said for Corollary~\ref{thm:univariate-catalyst-embedding}, for a catalyst $v$ also $Q v$ will be a catalyst for any unitary $Q$, since unitary transformations will preserve~(\ref{eq:unidir-gamma-bound-constraint-partial-l2}). Therefore, up to such unitary transformation we can rewrite $v$ as a direct sum
\begin{align*}
    Q^\dag v = v^{(0)} \oplus v^{(1)} \oplus \cdots \oplus v^{(n-1)}
\end{align*}
where $v^{(k)}$ is a scalar polynomial of degree $\le k$. We say that such a catalyst is in \emph{triangular form}, and any catalyst can always be reduced to such form through a $QL$-decomposition of the matrix containing the Fourier coefficients $v_0, \ldots, v_{n-1}$ as columns. The $k$-th row of the lower triangular matrix will contain the coefficients of $v^{(k)}$. Note that the triangular form is not necessarily unique, as it is sufficient to multiply by a diagonal unitary (i.e., multiply each $v^{(k)}$ by some phase), to get a different catalyst in triangular form.
}

{
\begin{theorem}
    \label{thm:gamma-bound-to-gqsp}
    Let $P, Q$ be two polynomials of degree $n$ satisfying $|P(z)|^2 + |Q(z)|^2 = 1$ on $\T$. There exists a bijective mapping between the catalysts in triangular form of the adversary bound for $\ket{0} \mapsto \tau(z) = P(z) \ket{0} + Q(z) \ket{1}$ state conversion with oracle $O(z) = z$ and the QSP protocols over $SU(2)$ implementing $(P, Q)$. The catalyst pieces $v^{(1)}, \ldots, v^{(n-1)}$ will appear as polynomials in the intermediate step of the QSP protocol.
\end{theorem}
Another interpretation of this statement is that there is a correspondence between the space of catalysts and the set of QSP protocols.
\begin{proof}[Proof Sketch]
    For a catalyst in triangular form $v = v^{(0)} \oplus \cdots \oplus v^{(n-1)}$, we simply show that $(P(z), Q(z))$ can be reduced to $(v^{(n-1)}(z), \cdot)$ through layer stripping, i.e., we can find $A_n \in SU(2)$ such that
    \begin{align*}
        \tau(z) \stackrel{\Tilde{w}^\dag A_n^\dag}{\mapsto} v^{(n-1)}(z) \oplus s(z)
    \end{align*}
    for some scalar polynomial $s(z)$. By induction this also implies that $v^{(k)}(z)$ will satisfy
    \begin{align*}
        \bra{0} A_{k-1} \Tilde{w} \cdots \Tilde{w} A_0 \ket{0} = v^{(k)}(z) \ ,
    \end{align*}
    i.e., the $v^{(k)}$ will be the polynomials appearing in the intermediate steps of the QSP protocol we construct.
\end{proof}}
\noindent The decomposition carried out by Theorem~\ref{thm:gamma-bound-to-gqsp} uses the catalyst differently compared to Algorithm~\ref{alg:state-conversion}. In the latter case, the catalyst is left unprocessed, we apply the same unitary $T$ times, and $v$ is left unchanged in its half of the space (hence the name \emph{catalyst}). This requires a number of qubits that scales logarithmically with the rank of the catalyst. On the other hand, Theorem~\ref{thm:gamma-bound-to-gqsp} splits $v$ into polynomials of increasing degrees, which will give directions on how the layer stripping argument should proceed, and preserves the $SU(2)$-structure.

By this argument we can infer some structure for the catalysts.
\begin{corollary}
    {If $v(z) = v^{(0)} \oplus \cdots \oplus v^{(n-1)}$ is in triangular form}, then
    $$\abs*{v^{(k)}(z)} \le 1$$
    for every $z \in \T$, and hence $\norm{v(z)}^2 \le n$. Since unitary transformations preserve the Las Vegas complexity, then any catalyst for polynomial state conversion $\ket{0} \mapsto \tau(z)$ with a  $n$-degree $\tau$ has Las Vegas complexity $\le n$.
\end{corollary}
\noindent The corollary comes from the fact that $v^{(k)}$ is a component of the state in the QSP protocol at step $k$. The second claim is not surprising, since the Las Vegas complexity is at most the Monte Carlo complexity, which is equal to the length of the QSP protocol, {so it can be seen as a direct implication of Theorem~\ref{thm:gqsp-theorem}}.

Solving the semidefinite program for $v(z)$ is easy, as it simply requires to solve the linear system given by~(\ref{eq:unidir-gamma-bound-constraint-partial-l2-fourier}):
\begin{align}
    \delta_{k,0} \delta_{h,0} - \langle \tau_k, \tau_h \rangle & = X_{k,h} - X_{k-1, h-1} & \text{for every $k, h \in \qty{0, \ldots, n-1}$}
    \label{eq:constraint-univariate-to-qsp-fourier}
\end{align}
where $\delta$ is the Kronecker delta, and $X$ is a $n \times n$ matrix (we implicitly set $X_{k,h} = 0$ whenever $(k,h) \not\in \{ 0, \ldots, n-1 \}^2$). There is a unique $X$ that solves this system, whose entries can be computed as
\begin{align}
    X_{k,h} & = \sum_{j = 0}^{\infty} \delta_{k-j,0} \delta_{h-j,0} - \langle \tau_{k-j}, \tau_{h-j} \rangle \label{eq:adversary-qsp-unique-gram-matrix}
\end{align}
where the infinite sum is only a simplifying notation, as the support of $\tau$ is finite. By uniqueness, we conclude $X \succeq 0$, otherwise there would not exist a catalyst, which in turn contradicts the existence of a QSP protocol given by Theorem~\ref{thm:gqsp-theorem}. In order to obtain the catalyst $v(z) = v^{(0)} \oplus \cdots \oplus v^{(n-1)}$ we simply do a Cholesky decomposition {$X = L^\dag L$ (with $L$ being lower triangular, the $k$-th row will give the coefficients for $v^{(k)}$)}.
\begin{corollary}
    \label{thm:univariate-qsp-catalyst-uniqueness}
    The catalyst $v$ for univariate QSP $\ket{0} \mapsto P(z) \ket{0} + Q(z) \ket{1}$ is unique up to a unitary transformation.
\end{corollary}
\noindent This result can also be obtained from the bijectivity of the non-Linear Fourier transform~\cite{alexisQuantumSignalProcessing2024,alexisInfiniteQuantumSignal2024}, along with the catalyst $\leftrightarrow$ QSP protocol correspondence given by Theorem~\ref{thm:gamma-bound-to-gqsp} (by multiplying $v$ with a diagonal unitary, we obtain non-canonical QSP protocols, in the sense of~\cite[Corollary~10]{laneveGeneralizedQuantumSignal2025}).

\subsection{More than two dimensions and partial state conversion}
\label{sec:qsp-sun-as-state-conversion}
The partial version of the $\gamma_2$-bound would give us, in principle, also a way to find a complementary polynomial $Q(z)$. However, a feasible solution $(G_{\sigma}, v)$ might have the Gram operator $G_{\sigma}$ of the complementary state $\sigma(z)$ with rank $r > 1$, which implies having $$\sigma(z) = Q_1(z) \oplus \ldots \oplus Q_r(z)$$ satisfying the normalization $|P|^2 + \sum_{j} |Q_j|^2 = 1$ on $\T$.

We thus need to include QSP over $SU(r+1)$ with $r > 1$ in order to have the big picture. Considering a total state generation problem with the target state $\tau(z)$ being an $(r+1)$-dimensional state $P(z) \oplus Q_1(z) \oplus \cdots \oplus Q_r(z)$, we obtain a similar result for general $r$.

{
\begin{theorem}
    \label{thm:gamma-bound-to-qsp-sun}
    Let $r \ge 1$, and let $P, Q_1, \ldots, Q_r$ be polynomials of degree $\le n$ satisfying $|P(z)|^2 + \sum_{j=1}^r |Q_j(z)|^2 = 1$ on $\T$. There exists a bijective mapping between the catalysts in triangular form of the adversary bound for the state conversion problem $$\ket{0} \mapsto \tau(z) = P(z) \ket{0} + \sum_{j = 1}^r Q_j(z) \ket{j}$$ with oracle $O(z) = z$ and the QSP protocols over $SU(r+1)$ implementing $(P, Q_1, \ldots, Q_r)$.
\end{theorem}}
\begin{proof}[Proof Sketch]
    Exactly like in Theorem~\ref{thm:gamma-bound-to-gqsp}, the $v^{(k)}$'s give the intermediate steps for the layer stripping process.
\end{proof}
As Theorem~\ref{thm:qsp-sun-theorem} is a straight generalization of Theorem~\ref{thm:gqsp-theorem}, this extends the QSP $\leftrightarrow$ catalyst space bijection to the $r > 1$ case. This allows us to give a full characterization of univariate QSP through the adversary bound for partial state conversion, taking into account both polynomial completion and phase factor computation.
\begin{corollary}
    {There is a bijection between the feasible solution $(G_{\sigma}, v)$ to the adversary bound for a partial state conversion $\ket{0} \mapsto P(z) \ket{0}$ with $v$ being in triangular form, and any possible QSP protocol constructing $P(z)$.} The protocol will be over $SU(r+1)$, where $r$ is the rank of $G_{\sigma}$.
\end{corollary}
We are thus interested in minimal-rank solutions to the adversary bound to identify minimal-space QSP protocols. We can infer something more on this solution space.
\begin{lemma}
    Suppose $\sigma_1(z), \sigma_2(z)$ are two complementary states to $\tau(z)$. Then any superposition $\alpha \sigma_1 \oplus \beta \sigma_2$ is also a complementary state.
\end{lemma}
\begin{proof}
    Being complementary means that $\norm{\sigma_1}^2 = \norm{\sigma_2}^2 = 1 - \norm{\tau}^2$, this is true also for the convex combination $\abs{\alpha}^2 \norm{\sigma_1}^2 + \abs{\beta}^2 \norm{\sigma_2}^2$.
\end{proof}
In other words, by convexity of the adversary bound, if we have two feasible solutions $(G_{\sigma_1}, v_1), (G_{\sigma_2}, v_2)$, then the catalyst $\alpha v_1 \oplus \beta v_2$ {(whose Gram matrix is the convex combination $\abs{\alpha}^2 G_{v_1} + \abs{\beta}^2 G_{v_2}$)} identifies a QSP protocol constructing $\tau \oplus \alpha \sigma_1 \oplus \beta \sigma_2$ {(where this new complementary state has Gram matrix $\abs{\alpha}^2 G_{\sigma_1} + \abs{\beta}^2 G_{\sigma_2}$)}\footnote{We remark that this direct sum does not mean that the dimension of the new QSP protocol will be the sum of the dimensions of the original protocols, essentially because the rank of the linear combination of two matrices might be smaller than the sum of the ranks {(consider, e.g., the case $\sigma_1 = \sigma_2$)}.}. This leads to the following statement:
\begin{theorem}
    For a partial state conversion $\ket{0} \mapsto P(z) \ket{0}$, the $SU(2)$-QSP protocols $\calP_A(z)$ {satisfying $\bra{0} \calP_A(z) \ket{0} = P(z)$} form a convex hull for a space of QSP protocols {over arbitrary dimensions satisfying the same constraint}.
\end{theorem}
Whether this convex hull actually contains \emph{all} the possible QSP protocols implementing $P$ is unclear. More specifically, if $\sigma(z)$ is a rank-$r$ complementary state to $P(z)$, can we always write it down as a combination of $\sigma = \alpha_1 Q_1 \oplus \cdots \oplus \alpha_r Q_r$ where each $Q_k$ is a rank-$1$ complementary polynomial? If this is true, then
\begin{align*}
    P \oplus \sigma = \alpha_1 (P \oplus Q_1) \oplus \alpha_2 (P \oplus Q_2) \oplus \cdots \oplus \alpha_r (P \oplus Q_r)
\end{align*}
up to a unitary transformation, which would correspond to $r$ parallel $SU(2)$-QSP protocols (whose dimension can be possibly reduced using Theorem~\ref{thm:gamma-bound-to-qsp-sun}).



\subsection{Polynomial state conversion over multiple variables}
\label{sec:mqsp-as-state-conversion}
After building a theory over polynomials of a single variable, the results show that the feasible solution space of the adversary bound perfectly coincides with the space of all the possible QSP protocols. This makes the adversary bound a good candidate for filling the gap in the theory of M-QSP~\cite{rossiMultivariableQuantumSignal2022,nemethVariantsMultivariateQuantum2023,moriCommentMultivariableQuantum2024,gomesMultivariableQSPBosonic2024,laneveMultivariatePolynomialsAchievable2025}, since the catalyst gave directions for the layer stripping process. While this `hint' is not necessary in the univariate case, we'll see that it is precious for M-QSP, as the existence of a valid catalyst constitutes a sufficient condition for the existence of a M-QSP protocol. We consider the oracle $$O(\vv{z}) = \diag(z_1, \ldots, z_m) =: \sum_{j=1}^m z_j \Pi_j$$ where $\Pi_j = \ketbra{j}{j}$\footnote{{Like $O$, we will also use $\Pi_j$ to actually mean $\Pi_j \otimes \id_{\calW}$ when applied to a vector in $\calM \otimes \calW$.}}. By taking the Fourier coefficients $e_{\vv{k},\vv{h}}, v_{\vv{k}}$ of $e(\vv{x}, \vv{y}), v(\vv{z})$, (\ref{eq:unidir-gamma-bound-constraint-partial-l2}) gives an analogous form to (\ref{eq:unidir-gamma-bound-constraint-partial-l2-fourier}):
\begin{align}
    e_{\vv{k},\vv{h}} & = \langle v_{\vv{k}}, v_{\vv{h}} \rangle - \sum_{j = 1}^m \langle v_{\vv{k}-\vv{e}_j}, \Pi_j v_{\vv{h}-\vv{e}_j} \rangle \ .
    \label{eq:gamma-bound-constraint-multivariate-fourier}
\end{align}
where $\vv{e}_j$ is the $j$-th element of the standard basis {(i.e., $\vv{k} - \vv{e}_j$ subtracts $1$ from the $j$-th position of $\vv{k}$)}. In other words, we are implicitly dividing the catalyst into multiple catalysts $v = \bigoplus_j \Pi_j v$. We obtain a very similar result as in the univariate case (full proofs in Appendix~\ref{apx:catalyst-proofs-multivariate}).
\begin{lemma}
    \label{thm:gamma-bound-finite-polynomial-multivariate}
    Suppose the oracle is $O(\vv{z}) = \diag(\vv{z})$ and that $\xi(\vv{z}), \tau(\vv{z})$ have degree $\le n_j$ in $z_j$ for every $j$. Then any feasible solution $v(\vv{z})$ to the total polynomial state conversion $\xi \mapsto \tau$ must be a polynomial of degree $\le n_j - 1$ in each variable $z_j$.
\end{lemma}
\begin{proof}[Proof Sketch]
    Consider the hyperrectangle $\calS = \{ \vv{h} \in \Z^m : 0 \le h_j < n_j \}$. As in Lemma~\ref{thm:feasible-space-polynomial-univariate}, we claim that if $v$ is supported outside of $\calS$, then we can find a divergent subsequence and apply Bessel's inequality to contradict $v \in L^2$.
\end{proof}
We thus know that $v(\vv{z})$ will follow the degree of $\xi, \tau$ just like in the univariate case. We now would like to apply a unitary {to obtain a triangular form}. However, not every unitary transformation $v \mapsto Qv$ will preserve (\ref{eq:gamma-bound-constraint-multivariate-fourier}) in this case, but only the ones that commute with the oracle $O(\vv{z})$ (in the previous sections, any unitary commutes with $z$). In other words, $U$ should be written as a direct sum (block-diagonal) $U_1 \oplus \cdots \oplus U_m$, transforming each component $\Pi_j v \mapsto U_j \Pi_j v$ independently.

We now use each $U_j$ to express $\Pi_j v = \Pi_j v^{(0)} \oplus \cdots \oplus \Pi_j v^{(n-1)}$. Thus, the subcatalyst $v^{(k)} := \bigoplus_j \Pi_j v^{(k)}$ is a polynomial vector of degree $\le k$. {Analogously to the univariate case we can say that $v$ written as 
\begin{align*}
    v = v^{(0)} \oplus v^{(1)} \oplus \cdots \oplus v^{(n-1)}
\end{align*}
is thus in triangular form.} An important difference here is that, while with $m = 1$ we proved that each $v^{(k)}$ is one-dimensional, we do not have guarantees on the dimensions of $v^{(k)}$ here. This has connections with the non-universality of M-QSP over $SU(2)$~\cite{rossiMultivariableQuantumSignal2022,
moriCommentMultivariableQuantum2024}: for $v$ to represent a M-QSP protocol over $SU(2)$ in the sense of~(\ref{eq:rossi-chuang-qsp}), for each $k$ we must have exactly one $j$ such that $\Pi_j v^{(k)} \neq 0$ (that $j$ gives the signal operator to apply at step $k$).
\begin{theorem}
    \label{thm:gamma-bound-to-mqsp}
    Consider the total state conversion $$\ket{0} \mapsto \tau(\vv{z}) = P(\vv{z}) \ket{0} + \sum_{k = 1}^r Q_k(\vv{z}) \ket{k} \ .$$
    where $P, Q_k$ have degree $\le n$. {A catalyst in triangular form of the adversary bound for the above state conversion problem can be turned into a M-QSP protocol of length $n$ for $\tau(\vv{z})$.}
\end{theorem}
{}
\begin{proof}[Proof Sketch]
    Exactly like in Theorem~\ref{thm:gamma-bound-to-qsp-sun}, the subcatalysts $v^{(k)}$ tell us how to carry out the layer stripping process. In particular, they guarantee the success of the induction step.
\end{proof}
More compactly, it is possible to characterize the set of M-QSP protocols geometrically.
\begin{corollary}
    Let $\ket{0} \mapsto P(\vv{z}) \ket{0}$ be a partial state conversion problem. The following subset of the semidefinite cone identifies the set of all M-QSP protocols of length $\le \ell$ achieving $P(\vv{z})$:
    \begin{align*}
        \calQ^{\ell}_{P} := \qty{(X^{(1)}, \ldots, X^{(m)}) \in (\calS^{\ell^m})^m : G_0 - G_P \succeq \sum_{j = 1}^m X^{(j)}_{\vv{k}, \vv{h}} - X^{(j)}_{\vv{k} - \vv{e}_j, \vv{h} - \vv{e}_j}, \ X^{(j)} \succeq 0 }
    \end{align*}
    where $\calS^{\ell^m}$ is the cone of $\ell^m \times \ell^m$ positive semidefinite matrices, $G_0, G_P$ are the Gram matrices of $\ket{0}$ and $P(\vv{z})$ respectively, expressed in the Fourier basis, while $X^{(j)}$ is the Gram matrix of the $j$-th catalyst $\Pi_j v$.
\end{corollary}
The difference $G_0 - G_P - \sum_{j = 1}^m X^{(j)}_{\vv{k}, \vv{h}} - X^{(j)}_{\vv{k} - \vv{e}_j, \vv{h} - \vv{e}_j} \succeq 0$ is simply the Gram matrix $G_{\sigma}$ of the complementary polynomials. $\ell$ is needed as the multivariate Fej{\'e}r-Riesz theorem (Theorem~\ref{thm:fejer-riesz-multivariate}) does not guarantee an upper bound for the degrees of the complementary polynomials. We define the set of all M-QSP protocols for $P$ as $\calQ_P := \bigcup_{\ell > 0} \calQ^\ell_P$. We can of course replace $P$ with a total or partial state $\tau(\vv{z})$ of any dimension, thus defining $\calQ^{\ell}_\tau, \calQ_\tau$ analogously.

Unfortunately, Theorem~\ref{thm:gamma-bound-to-mqsp} cannot say anything about the dimensionality of the protocol, unlike its univariate counterpart. As an example, consider the one-dimensional total polynomial state $\tau(\vv{z}) = z_1 z_2 \cdots z_m$: the following protocol needs $m$ dimensions
\begin{align}
    \frac{1}{\sqrt{m}}
    \begin{bmatrix}
        1 \\ 1 \\ \vdots \\ 1
    \end{bmatrix}
    \mapsto
    \frac{1}{\sqrt{m}}
    \begin{bmatrix}
        z_1 \\ z_2 \\ \vdots \\ z_m
    \end{bmatrix}
    \mapsto
    \cdots
    \mapsto
    \frac{1}{\sqrt{m}}
    \begin{bmatrix}
        z_1 z_2 \cdots z_{m-1} \\ z_2 z_3 \cdots z_m \\ \vdots \\ z_m z_1 \cdots z_{m-2}
    \end{bmatrix}
    \mapsto
    \frac{1}{\sqrt{m}}
    \begin{bmatrix}
        z_1 z_2 \cdots z_m \\ z_1 z_2 \cdots z_m \\ \vdots \\ z_1 z_2 \cdots z_m
    \end{bmatrix}
    \label{eq:mqsp-dimensionality-counterexample}
\end{align}
where the last state is equal to $\tau$ up to a unitary transformation (essentially the intermediate processing operators are chosen to be the cyclic shift matrix). Nonetheless, it is possible to construct the same state using only one dimension, by multiplying by one $z_k$ at each step. It is possible to see the above protocol as a superposition of instances of this one-dimensional protocol over different sequences of $z_k$. In terms of the semidefinite formulation, define the catalyst matrix $X_{\pi} \in \calQ^m_{\tau}$ to be the one identifying the M-QSP protocol which applies the $z_j$ in the order defined by the permutation {$\pi : \{ 1, \ldots, m \} \rightarrow \{ 1, \ldots, m \}$}. Then, the catalyst matrix $X \in \calQ^m_{\tau}$ representing the protocol in (\ref{eq:mqsp-dimensionality-counterexample}) is a convex combination of the $X_{\pi}$. We do not know whether the existence of these minimal-space catalyst matrices is guaranteed, in particular we leave
\begin{conjecture}[Expanded from~\cite{laneveMultivariatePolynomialsAchievable2025}]
    \label{thm:mqsp-dimensionality-conjecture}
    Every $\tau(\vv{z})$ of dimension $d$ constructible by M-QSP allows a M-QSP protocol over $SU(d)$.
\end{conjecture}
The stronger guarantee on the dimension in the univariate case given by Theorem~\ref{thm:gamma-bound-to-qsp-sun} relies on the uniqueness of the catalyst matrix $X$ in~(\ref{eq:constraint-univariate-to-qsp-fourier}), while the example in~(\ref{eq:mqsp-dimensionality-counterexample}) clearly shows distinct possible solutions (which in turn includes also any convex combination).

It is worth mentioning that previous impossibility proofs in~\cite{nemethVariantsMultivariateQuantum2023,laneveMultivariatePolynomialsAchievable2025} only talk about total states in a specific dimension. In other words, their arguments do not imply $\calQ_{\tau} = \emptyset$, and not even $\calQ_P = \emptyset$ or $\calQ_Q = \emptyset$.

\section{Discussion and outlook}
\label{sec:discussion}

This work uses the adversary bound --- a central tool in query complexity --- to show the theory of quantum signal processing (QSP) from a new perspective. This new approach --- which recasts QSP as a special instance of state conversion --- allows one to better understand the structure of QSP, especially in the multivariate case. In particular, we found that the catalysts given by the adversary bound can always be used to construct an optimal-space protocol in the case of univariate oracle $O(z) = z$, unlike the algorithm of~\cite{belovsOneWayTicketVegas2023}, whose required space always depends on the dimensionality of the catalyst. Surprisingly, the constraint (\ref{eq:constraint-univariate-to-qsp-fourier}) of the adversary bound for the univariate case shows a displacement structure similar to the one in~\cite{niFastPhaseFactor2024,niInverseNonlinearFast2025}.

The partial state conversion formulation is convenient, as the adversary bound encodes both the problem of finding a complementary polynomial and computing the corresponding QSP protocol. Dimensionality might still be a problem in the multivariate case, confirming that univariate QSP is a very special and well-behaved case of state conversion. We thus turn our attention to the following outlook:
\begin{enumerate}[(i)]
    \item Finding a minimal-space M-QSP protocol implementing some desired transformation $P(\vv{z})$ requires finding minimal-rank elements of $\calQ_P$. While rank minimization problems are usually NP-hard, might it be possible to at least approximate minimal rank with convex heuristics~\cite{rechtNecessarySufficientConditions2008}?
    
    \item Theoretically any $P$ with $|P| < 1$ can be achieved by the M-QSP over multiple dimensions we presented: one can think of a linear combination of unitaries~\cite{berryHamiltonianSimulationNearly2015} boosted with sufficiently many rounds of amplitude amplification. This naive approach, however, requires many qubits, and the adversary bound approach we presented would be able to give further optimization. In particular, can such $P$ admit a solution whose rank depends only on the number $m$ of variables, or the dimension $d$ of the full, completed polynomial vector (as in Conjecture~\ref{thm:mqsp-dimensionality-conjecture})?
    
    \item What happens when the degree $n$ of $\tau$ goes to infinity? This is the setting of infinite quantum signal processing~\cite{dongInfiniteQuantumSignal2024,alexisInfiniteQuantumSignal2024}, and the adversary bound should admit feasible catalysts of infinite rank. It is not clear, however, if the infinite-rank solution achieving a finite objective function would remain unique in the univariate case, as in Corollary~\ref{thm:univariate-qsp-catalyst-uniqueness}.
    
    \item Las Vegas complexity satisfies exactness and thriftiness properties, allowing to having complexity-preserving compositions. In light of this work, would there be a connection between the transducer formalism of~\cite{belovsTamingQuantumTime2024} and the gadgets of~\cite{rossiModularQuantumSignal2025}? This would require a translation of the latter work into the analytic QSP picture;
    
    \item In this work we manage to find bounds for the space used by the protocols in the case of polynomial state conversion. Can we derive space bounds also for finite-labels state conversion~\cite{leeQuantumQueryComplexity2011,belovsOneWayTicketVegas2023}? The shape of the oracle $O(z)$ in our case exhibits a particular structure, thus suggesting that this plays an important role. But there might be a way to split and consume the catalyst as we did in this work, so that embedding the full catalyst in the registers like in~\cite{belovsOneWayTicketVegas2023} might be unnecessary.
    
    \item Importantly, QSP is usually lifted to quantum eigenvalue transformation of unitaries (QET-U)~\cite{dongGroundStatePreparationEnergy2022} or quantum singular value transformation (QSVT)~\cite{gilyenQuantumSingularValue2019} for algorithmic applications. While a lift for the former is clear in the M-QSP setting over $SU(N)$ (just replace the occurrences of $z_j$ in the oracles with controlled versions of the unitaries $U_j$), it is unclear how one can obtain a multivariate QSVT from it. The authors in~\cite{rossiMultivariableQuantumSignal2022} provide a M-QSVT ansatz for their M-QSP over $SU(2)$, but the single-qubit structure is important for their construction. More recent work on obtaining a QSVT ansatz from GQSP might be relevant for this direction~\cite{sunderhaufGeneralizedQuantumSingular2023}.

    \item The adversary bound formalism here does not seem to have fundamental limitations if we remove the assumption that the $z_j$'s do not commute. This \emph{non-abelian} version of M-QSP adds rigidity to the structure of the polynomials. We suspect that the non-uniqueness of the catalysts for M-QSP follows from the fact that the polynomials $z_1 z_2$ and $z_2 z_1$ are the same (which is the gist of the counterexample in~(\ref{eq:mqsp-dimensionality-counterexample})). This leads to the belief that non-abelian M-QSP might be easier to analyze than its abelian counterpart.
\end{enumerate}

\section*{Acknowledgements}
I would like to thank Hitomi Mori and Yuki Ito for discussions on M-QSP, as well as Alexander Belov and Troy Lee for important clarifications on state conversion and query complexity. This work is supported by the Swiss National Science Foundation (SNSF), project No.\ 200020-214808.

\bibliographystyle{quantum}
\bibliography{refs}

\begin{thebibliography}{10}

\bibitem{brassardQuantumAmplitudeAmplification2002}
Gilles Brassard, Peter Hoyer, Michele Mosca, and Alain Tapp.
\newblock ``Quantum {{Amplitude Amplification}} and {{Estimation}}''.
\newblock \href{https://dx.doi.org/10.1090/conm/305/05215}{Quantum Computation and Information {\bf 305}, 53--74}~(2002).

\bibitem{berryHamiltonianSimulationNearly2015}
Dominic~W Berry, Andrew~M Childs, and Robin Kothari.
\newblock ``Hamiltonian {{Simulation}} with {{Nearly Optimal Dependence}} on all {{Parameters}}''.
\newblock In 2015 {{IEEE}} 56th {{Annual Symposium}} on {{Foundations}} of {{Computer Science}}.
\newblock \href{https://dx.doi.org/10.1109/FOCS.2015.54}{Pages 792--809}.
\newblock ~(2015).

\bibitem{childsHamiltonianSimulationUsing2012}
Andrew~M Childs and Nathan Wiebe.
\newblock ``Hamiltonian simulation using linear combinations of unitary operations''.
\newblock \href{https://dx.doi.org/10.26421/QIC12.11-12-1}{Quantum Information and Computation {\bf 12}, 901--924}~(2012).

\bibitem{anLinearCombinationHamiltonian2023}
Dong An, Jin-Peng Liu, and Lin Lin.
\newblock ``Linear {{Combination}} of {{Hamiltonian Simulation}} for {{Nonunitary Dynamics}} with {{Optimal State Preparation Cost}}''.
\newblock \href{https://dx.doi.org/10.1103/PhysRevLett.131.150603}{Physical Review Letters {\bf 131}, 150603}~(2023).

\bibitem{anLaplaceTransformBased2024}
Dong An, Andrew~M. Childs, Lin Lin, and Lexing Ying.
\newblock ``Laplace transform based quantum eigenvalue transformation via linear combination of {{Hamiltonian}} simulation''~(2024).
\newblock  \href{http://arxiv.org/abs/2411.04010}{arXiv:2411.04010}.

\bibitem{nielsenQuantumComputationQuantum2010a}
Michael~A. Nielsen and Isaac~L. Chuang.
\newblock ``Quantum computation and quantum information: 10th anniversary edition''.
\newblock Cambridge University Press. Cambridge~(2010).
\newblock  url:~\url{https://doi.org/10.1017/CBO9780511976667}.

\bibitem{harrowQuantumAlgorithmLinear2009}
Aram~W. Harrow, Avinatan Hassidim, and Seth Lloyd.
\newblock ``Quantum {{Algorithm}} for {{Linear Systems}} of {{Equations}}''.
\newblock \href{https://dx.doi.org/10.1103/PhysRevLett.103.150502}{Physical Review Letters {\bf 103}, 150502}~(2009).

\bibitem{aharonovQuantumWalksGraphs2001}
Dorit Aharonov, Andris Ambainis, Julia Kempe, and Umesh Vazirani.
\newblock ``Quantum walks on graphs''.
\newblock In Proceedings of the Thirty-Third Annual {{ACM}} Symposium on {{Theory}} of Computing.
\newblock \href{https://dx.doi.org/10.1145/380752.380758}{Pages 50--59}.
\newblock {{STOC}} '01. Association for Computing Machinery~(2001).

\bibitem{ambainisQuantumWalkAlgorithm2004}
A~Ambainis.
\newblock ``Quantum walk algorithm for element distinctness''.
\newblock In 45th {{Annual IEEE Symposium}} on {{Foundations}} of {{Computer Science}}.
\newblock \href{https://dx.doi.org/10.1109/FOCS.2004.54}{Pages 22--31}.
\newblock ~(2004).

\bibitem{szegedyQuantumSpeedMarkovChain2004}
Mario Szegedy.
\newblock ``Quantum speed-up of {{Markov}} chain based algorithms''.
\newblock In 45th {{Annual IEEE Symposium}} on {{Foundations}} of {{Computer Science}}.
\newblock \href{https://dx.doi.org/10.1109/FOCS.2004.53}{Pages 32--41}.
\newblock ~(2004).

\bibitem{belovsQuantumWalksElectric2013}
Aleksandrs Belovs.
\newblock ``Quantum {{Walks}} and {{Electric Networks}}''~(2013).
\newblock  \href{http://arxiv.org/abs/1302.3143}{arXiv:1302.3143}.

\bibitem{apersUnifiedFrameworkQuantum2021}
Simon Apers, Andr{\'a}s Gily{\'e}n, and Stacey Jeffery.
\newblock ``A {{Unified Framework}} of {{Quantum Walk Search}}''.
\newblock In Leibniz {{International Proceedings}} in {{Informatics}} ({{LIPIcs}}).
\newblock \href{https://dx.doi.org/10.4230/LIPIcs.STACS.2021.6}{Volume 187 of Leibniz {{International Proceedings}} in {{Informatics}} ({{LIPIcs}}), pages 6:1--6:13}.
\newblock Schloss Dagstuhl -- Leibniz-Zentrum f\"ur Informatik~(2021).

\bibitem{gilyenQuantumSingularValue2019}
Andr{\'a}s Gily{\'e}n, Yuan Su, Guang~Hao Low, and Nathan Wiebe.
\newblock ``Quantum singular value transformation and beyond: Exponential improvements for quantum matrix arithmetics''.
\newblock In Proceedings of the 51st {{Annual ACM SIGACT Symposium}} on {{Theory}} of {{Computing}}.
\newblock \href{https://dx.doi.org/10.1145/3313276.3316366}{Pages 193--204}.
\newblock ACM~(2019).

\bibitem{lowHamiltonianSimulationUniform2017}
Guang~Hao Low and Isaac~L Chuang.
\newblock ``Hamiltonian {{Simulation}} by {{Uniform Spectral Amplification}}''~(2017).
\newblock  \href{http://arxiv.org/abs/1707.05391}{arXiv:1707.05391}.

\bibitem{lowOptimalHamiltonianSimulation2017}
Guang~Hao Low and Isaac~L. Chuang.
\newblock ``Optimal {{Hamiltonian Simulation}} by {{Quantum Signal Processing}}''.
\newblock \href{https://dx.doi.org/10.1103/PhysRevLett.118.010501}{Physical Review Letters {\bf 118}, 010501}~(2017).

\bibitem{lowQuantumSignalProcessing2017}
Guang~Hao Low.
\newblock ``Quantum signal processing by single-qubit dynamics''.
\newblock Thesis.
\newblock Massachusetts Institute of Technology.
\newblock ~(2017).
\newblock  url:~\url{https://dspace.mit.edu/handle/1721.1/115025}.

\bibitem{lowHamiltonianSimulationQubitization2019}
Guang~Hao Low and Isaac~L. Chuang.
\newblock ``Hamiltonian {{Simulation}} by {{Qubitization}}''.
\newblock \href{https://dx.doi.org/10.22331/q-2019-07-12-163}{Quantum {\bf 3}, 163}~(2019).

\bibitem{martynGrandUnificationQuantum2021}
John~M. Martyn, Zane~M. Rossi, Andrew~K. Tan, and Isaac~L. Chuang.
\newblock ``A {{Grand Unification}} of {{Quantum Algorithms}}''.
\newblock \href{https://dx.doi.org/10.1103/PRXQuantum.2.040203}{PRX Quantum {\bf 2}, 40203}~(2021).

\bibitem{berryDoublingEfficiencyHamiltonian2024}
Dominic~W. Berry, Danial Motlagh, Giacomo Pantaleoni, and Nathan Wiebe.
\newblock ``Doubling the efficiency of {{Hamiltonian}} simulation via generalized quantum signal processing''.
\newblock \href{https://dx.doi.org/10.1103/PhysRevA.110.012612}{Physical Review A {\bf 110}, 012612}~(2024).

\bibitem{haahProductDecompositionPeriodic2019}
Jeongwan Haah.
\newblock ``Product {{Decomposition}} of {{Periodic Functions}} in {{Quantum Signal Processing}}''.
\newblock \href{https://dx.doi.org/10.22331/q-2019-10-07-190}{Quantum {\bf 3}, 190}~(2019).

\bibitem{chaoFindingAnglesQuantum2020}
Rui Chao, Dawei Ding, Andras Gilyen, Cupjin Huang, and Mario Szegedy.
\newblock ``Finding {{Angles}} for {{Quantum Signal Processing}} with {{Machine Precision}}''~(2020).
\newblock  \href{http://arxiv.org/abs/2003.02831}{arXiv:2003.02831}.

\bibitem{dongEfficientPhasefactorEvaluation2021}
Yulong Dong, Xiang Meng, K~Birgitta Whaley, and Lin Lin.
\newblock ``Efficient phase-factor evaluation in quantum signal processing''.
\newblock \href{https://dx.doi.org/10.1103/PhysRevA.103.042419}{Physical Review A {\bf 103}, 42419}~(2021).

\bibitem{yingStableFactorizationPhase2022}
Lexing Ying.
\newblock ``Stable factorization for phase factors of quantum signal processing''.
\newblock \href{https://dx.doi.org/10.22331/q-2022-10-20-842}{Quantum {\bf 6}, 842}~(2022).

\bibitem{wangEnergyLandscapeSymmetric2022}
Jiasu Wang, Yulong Dong, and Lin Lin.
\newblock ``On the energy landscape of symmetric quantum signal processing''.
\newblock \href{https://dx.doi.org/10.22331/q-2022-11-03-850}{Quantum {\bf 6}, 850}~(2022).

\bibitem{dongRobustIterativeMethod2024}
Yulong Dong, Lin Lin, Hongkang Ni, and Jiasu Wang.
\newblock ``Robust {{Iterative Method}} for {{Symmetric Quantum Signal Processing}} in {{All Parameter Regimes}}''.
\newblock \href{https://dx.doi.org/10.1137/23M1598192}{SIAM Journal on Scientific Computing {\bf 46}, A2951--A2971}~(2024).

\bibitem{yamamotoRobustAngleFinding2024}
Shuntaro Yamamoto and Nobuyuki Yoshioka.
\newblock ``Robust {{Angle Finding}} for {{Generalized Quantum Signal Processing}}''~(2024).
\newblock  \href{http://arxiv.org/abs/2402.03016}{arXiv:2402.03016}.

\bibitem{berntsonComplementaryPolynomialsQuantum2025}
Bjorn~K. Berntson and Christoph S{\"u}nderhauf.
\newblock ``Complementary {{Polynomials}} in {{Quantum Signal Processing}}''.
\newblock \href{https://dx.doi.org/10.1007/s00220-025-05302-9}{Communications in Mathematical Physics {\bf 406}, 161}~(2025).

\bibitem{alexisInfiniteQuantumSignal2024}
Michel Alexis, Lin Lin, Gevorg Mnatsakanyan, Christoph Thiele, and Jiasu Wang.
\newblock ``Infinite quantum signal processing for arbitrary {{Szeg\"o}} functions''~(2024).
\newblock  \href{http://arxiv.org/abs/2407.05634}{arXiv:2407.05634}.

\bibitem{laneveGeneralizedQuantumSignal2025}
Lorenzo Laneve.
\newblock ``Generalized {{Quantum Signal Processing}} and {{Non-Linear Fourier Transform}} are equivalent''~(2025).
\newblock  \href{http://arxiv.org/abs/2503.03026}{arXiv:2503.03026}.

\bibitem{niFastPhaseFactor2024}
Hongkang Ni and Lexing Ying.
\newblock ``Fast {{Phase Factor Finding}} for {{Quantum Signal Processing}}''~(2024).
\newblock  \href{http://arxiv.org/abs/2410.06409}{arXiv:2410.06409}.

\bibitem{niInverseNonlinearFast2025}
Hongkang Ni, Rahul Sarkar, Lexing Ying, and Lin Lin.
\newblock ``Inverse nonlinear fast {{Fourier}} transform on {{SU}}(2) with applications to quantum signal processing''~(2025).
\newblock  \href{http://arxiv.org/abs/2505.12615}{arXiv:2505.12615}.

\bibitem{martynHalvingCostQuantum2025}
John~M. Martyn and Patrick Rall.
\newblock ``Halving the cost of quantum algorithms with randomization''.
\newblock \href{https://dx.doi.org/10.1038/s41534-025-01003-2}{npj Quantum Information {\bf 11}, 47}~(2025).

\bibitem{martynParallelQuantumSignal2024}
John~M. Martyn, Zane~M. Rossi, Kevin~Z. Cheng, Yuan Liu, and Isaac~L. Chuang.
\newblock ``Parallel {{Quantum Signal Processing Via Polynomial Factorization}}''~(2024).
\newblock  \href{http://arxiv.org/abs/2409.19043}{arXiv:2409.19043}.

\bibitem{mizutaRecursiveQuantumEigenvalue2024}
Kaoru Mizuta and Keisuke Fujii.
\newblock ``Recursive quantum eigenvalue and singular-value transformation: {{Analytic}} construction of matrix sign function by {{Newton}} iteration''.
\newblock \href{https://dx.doi.org/10.1103/PhysRevResearch.6.L012007}{Physical Review Research {\bf 6}, L012007}~(2024).

\bibitem{rossiModularQuantumSignal2025}
Zane~M. Rossi, Jack~L. Ceroni, and Isaac~L. Chuang.
\newblock ``Modular quantum signal processing in many variables''.
\newblock \href{https://dx.doi.org/10.22331/q-2025-06-18-1776}{Quantum {\bf 9}, 1776}~(2025).

\bibitem{rossiSolovayKitaevTheoremQuantum2025}
Zane~M. Rossi.
\newblock ``A {{Solovay-Kitaev}} theorem for quantum signal processing''~(2025).
\newblock  \href{http://arxiv.org/abs/2505.05468}{arXiv:2505.05468}.

\bibitem{rossiMultivariableQuantumSignal2022}
Zane~M. Rossi and Isaac~L. Chuang.
\newblock ``Multivariable quantum signal processing ({{M-QSP}}): Prophecies of the two-headed oracle''.
\newblock \href{https://dx.doi.org/10.22331/q-2022-09-20-811}{Quantum {\bf 6}, 811}~(2022).

\bibitem{moriCommentMultivariableQuantum2024}
Hitomi Mori, Kaoru Mizuta, and Keisuke Fujii.
\newblock ``Comment on "{{Multivariable}} quantum signal processing ({{M-QSP}}): Prophecies of the two-headed oracle"''.
\newblock \href{https://dx.doi.org/10.22331/q-2024-10-29-1512}{Quantum {\bf 8}, 1512}~(2024).

\bibitem{itoPolynomialTimeConstructive2024}
Yuki Ito, Hitomi Mori, Kazuki Sakamoto, and Keisuke Fujii.
\newblock ``Polynomial time constructive decision algorithm for multivariable quantum signal processing''~(2024).
\newblock  \href{http://arxiv.org/abs/2410.02332}{arXiv:2410.02332}.

\bibitem{nemethVariantsMultivariateQuantum2023}
Bal{\'a}zs N{\'e}meth, Blanka K{\"o}v{\'e}r, Bogl{\'a}rka Kulcs{\'a}r, Roland~Botond Mikl{\'o}si, and Andr{\'a}s Gily{\'e}n.
\newblock ``On variants of multivariate quantum signal processing and their characterizations''~(2023).
\newblock  \href{http://arxiv.org/abs/2312.09072}{arXiv:2312.09072}.

\bibitem{laneveMultivariatePolynomialsAchievable2025}
Lorenzo Laneve and Stefan Wolf.
\newblock ``On multivariate polynomials achievable with quantum signal processing''.
\newblock \href{https://dx.doi.org/10.22331/q-2025-02-20-1641}{Quantum {\bf 9}, 1641}~(2025).

\bibitem{borns-weilQuantumAlgorithmFunctions2023}
Yonah {Borns-Weil}, Tahsin Saffat, and Zachary Stier.
\newblock ``A {{Quantum Algorithm}} for {{Functions}} of {{Multiple Commuting Hermitian Matrices}}''~(2023).
\newblock  \href{http://arxiv.org/abs/2302.11139}{arXiv:2302.11139}.

\bibitem{gomesMultivariableQSPBosonic2024}
Niladri Gomes, Hokiat Lim, and Nathan Wiebe.
\newblock ``Multivariable {{QSP}} and {{Bosonic Quantum Simulation}} using {{Iterated Quantum Signal Processing}}''~(2024).
\newblock  \href{http://arxiv.org/abs/2408.03254}{arXiv:2408.03254}.

\bibitem{hoyerLowerBoundsQuantum2005}
Peter Hoyer and Robert Spalek.
\newblock ``Lower {{Bounds}} on {{Quantum Query Complexity}}''~(2005).
\newblock  \href{http://arxiv.org/abs/quant-ph/0509153}{arXiv:quant-ph/0509153}.

\bibitem{bealsQuantumLowerBounds2001}
Robert Beals, Harry Buhrman, Richard Cleve, Michele Mosca, and Ronald {de Wolf}.
\newblock ``Quantum lower bounds by polynomials''.
\newblock \href{https://dx.doi.org/10.1145/502090.502097}{J. ACM {\bf 48}, 778--797}~(2001).

\bibitem{ambainisPolynomialDegreeVs2004}
Andris Ambainis.
\newblock ``Polynomial degree vs. quantum query complexity''~(2004).
\newblock  \href{http://arxiv.org/abs/quant-ph/0305028}{arXiv:quant-ph/0305028}.

\bibitem{ambainisQuantumLowerBounds2000}
Andris Ambainis.
\newblock ``Quantum lower bounds by quantum arguments''.
\newblock In Proceedings of the Thirty-Second Annual {{ACM}} Symposium on {{Theory}} of Computing.
\newblock \href{https://dx.doi.org/10.1145/335305.335394}{Pages 636--643}.
\newblock {{STOC}} '00. Association for Computing Machinery~(2000).

\bibitem{spalekMultiplicativeQuantumAdversary2008}
Robert Spalek.
\newblock ``The {{Multiplicative Quantum Adversary}}''.
\newblock In Proceedings of the 2008 {{IEEE}} 23rd {{Annual Conference}} on {{Computational Complexity}}.
\newblock \href{https://dx.doi.org/10.1109/CCC.2008.9}{Pages 237--248}.
\newblock {{CCC}} '08. IEEE Computer Society~(2008).

\bibitem{spalekAllQuantumAdversary2005}
Robert Spalek and Mario Szegedy.
\newblock ``All quantum adversary methods are equivalent''.
\newblock In Proceedings of the 32nd International Conference on {{Automata}}, {{Languages}} and {{Programming}}.
\newblock \href{https://dx.doi.org/10.1007/11523468_105}{Pages 1299--1311}.
\newblock {{ICALP}}'05. Springer-Verlag~(2005).

\bibitem{belovsVariationsQuantumAdversary2015}
Aleksandrs Belovs.
\newblock ``Variations on {{Quantum Adversary}}''~(2015).
\newblock  \href{http://arxiv.org/abs/1504.06943}{arXiv:1504.06943}.

\bibitem{hoyerTightAdversaryBounds2006}
Peter Hoyer, Troy Lee, and Robert Spalek.
\newblock ``Tight adversary bounds for composite functions''~(2006).
\newblock  \href{http://arxiv.org/abs/quant-ph/0509067}{arXiv:quant-ph/0509067}.

\bibitem{montanaroQuantumClassicalQuery2024}
Ashley Montanaro and Changpeng Shao.
\newblock ``Quantum and {{Classical Query Complexities}} of {{Functions}} of {{Matrices}}''.
\newblock In Proceedings of the 56th {{Annual ACM Symposium}} on {{Theory}} of {{Computing}}.
\newblock \href{https://dx.doi.org/10.1145/3618260.3649665}{Pages 573--584}.
\newblock {{STOC}} 2024. Association for Computing Machinery~(2024).

\bibitem{hoyerNegativeWeightsMake2007}
Peter Hoyer, Troy Lee, and Robert Spalek.
\newblock ``Negative weights make adversaries stronger''.
\newblock In Proceedings of the Thirty-Ninth Annual {{ACM}} Symposium on {{Theory}} of Computing.
\newblock \href{https://dx.doi.org/10.1145/1250790.1250867}{Pages 526--535}.
\newblock ~(2007).

\bibitem{reichardtSpanProgramsQuantum2009}
Ben~W. Reichardt.
\newblock ``Span programs and quantum query complexity: {{The}} general adversary bound is nearly tight for every boolean function''.
\newblock In 2009 50th {{Annual IEEE Symposium}} on {{Foundations}} of {{Computer Science}}.
\newblock \href{https://dx.doi.org/10.1109/FOCS.2009.55}{Pages 544--551}.
\newblock ~(2009).

\bibitem{belovsDirectReductionPolynomial2024}
Aleksandrs Belovs.
\newblock ``A {{Direct Reduction}} from the {{Polynomial}} to the {{Adversary Method}}''.
\newblock In Fr{\'e}d{\'e}ric Magniez and Alex~Bredariol Grilo, editors, 19th {{Conference}} on the {{Theory}} of {{Quantum Computation}}, {{Communication}} and {{Cryptography}} ({{TQC}} 2024).
\newblock \href{https://dx.doi.org/10.4230/LIPIcs.TQC.2024.11}{Volume 310 of Leibniz {{International Proceedings}} in {{Informatics}} ({{LIPIcs}}), pages 11:1--11:15}.
\newblock Dagstuhl, Germany~(2024). Schloss Dagstuhl -- Leibniz-Zentrum f\"ur Informatik.

\bibitem{ambainisSymmetryassistedAdversariesQuantum2011}
Andris Ambainis, Lo{\"i}ck Magnin, Martin Roetteler, and J{\'e}r{\'e}mie Roland.
\newblock ``Symmetry-assisted adversaries for quantum state generation''.
\newblock In 2011 {{IEEE}} 26th {{Annual Conference}} on {{Computational Complexity}}.
\newblock \href{https://dx.doi.org/10.1109/CCC.2011.24}{Pages 167--177}.
\newblock ~(2011).

\bibitem{leeQuantumQueryComplexity2011}
Troy Lee, Rajat Mittal, Ben~W. Reichardt, Robert Spalek, and Mario Szegedy.
\newblock ``Quantum query complexity of state conversion''.
\newblock In 2011 {{IEEE}} 52nd {{Annual Symposium}} on {{Foundations}} of {{Computer Science}}.
\newblock \href{https://dx.doi.org/10.1109/FOCS.2011.75}{Pages 344--353}.
\newblock ~(2011).

\bibitem{belovsOneWayTicketVegas2023}
Aleksandrs Belovs and Duyal Yolcu.
\newblock ``One-{{Way Ticket}} to {{Las Vegas}} and the {{Quantum Adversary}}''~(2023).
\newblock  \href{http://arxiv.org/abs/2301.02003}{arXiv:2301.02003}.

\bibitem{belovsTamingQuantumTime2024}
Aleksandrs Belovs, Stacey Jeffery, and Duyal Yolcu.
\newblock ``Taming {{Quantum Time Complexity}}''.
\newblock \href{https://dx.doi.org/10.22331/q-2024-08-23-1444}{Quantum {\bf 8}, 1444}~(2024).

\bibitem{laneveQuantumSignalProcessing2024}
Lorenzo Laneve.
\newblock ``Quantum signal processing over {{SU}}({{N}})''~(2024).
\newblock  \href{http://arxiv.org/abs/2311.03949}{arXiv:2311.03949}.

\bibitem{luQuantumSignalProcessing2024}
Xi~Lu, Yuan Liu, and Hongwei Lin.
\newblock ``Quantum {{Signal Processing}} and {{Quantum Singular Value Transformation}} on {{U}}({{N}})''~(2024).
\newblock  \href{http://arxiv.org/abs/2408.01439}{arXiv:2408.01439}.

\bibitem{motlaghGeneralizedQuantumSignal2024}
Danial Motlagh and Nathan Wiebe.
\newblock ``Generalized {{Quantum Signal Processing}}''.
\newblock \href{https://dx.doi.org/10.1103/PRXQuantum.5.020368}{PRX Quantum {\bf 5}, 020368}~(2024).

\bibitem{sunderhaufGeneralizedQuantumSingular2023}
Christoph S{\"u}nderhauf.
\newblock ``Generalized {{Quantum Singular Value Transformation}}''~(2023).
\newblock  \href{http://arxiv.org/abs/2312.00723}{arXiv:2312.00723}.

\bibitem{hussenFejerRieszTheoremIts2021}
Abdulmtalb Hussen and Abdelbaset Zeyani.
\newblock ``Fejer-{{Riesz Theorem}} and {{Its Generalization}}''.
\newblock \href{https://dx.doi.org/10.29322/IJSRP.11.06.2021.p11437}{International Journal of Scientific and Research Publications (IJSRP) {\bf 11}, 286--292}~(2021).

\bibitem{geronimoPositiveExtensionsFejerRiesz2004}
Jeffrey~S. Geronimo and Hugo~J. Woerdeman.
\newblock ``Positive extensions, {{Fej\'er-Riesz}} factorization and autoregressive filters in two variables''.
\newblock \href{https://dx.doi.org/10.4007/ANNALS.2004.160.839}{Annals of Mathematics {\bf 160}, 839--906}~(2004).

\bibitem{dritschelFactorizationTrigonometricPolynomials2004}
Michael~A. Dritschel.
\newblock ``On {{Factorization}} of {{Trigonometric Polynomials}}''.
\newblock \href{https://dx.doi.org/10.1007/s00020-002-1198-4}{Integral Equations and Operator Theory {\bf 49}, 11--42}~(2004).

\bibitem{rudinExtensionProblemPositivedefinite1963}
Walter Rudin.
\newblock ``The extension problem for positive-definite functions''.
\newblock \href{https://dx.doi.org/10.1215/ijm/1255644960}{Illinois Journal of Mathematics {\bf 7}, 532--539}~(1963).

\bibitem{alexisQuantumSignalProcessing2024}
Michel Alexis, Gevorg Mnatsakanyan, and Christoph Thiele.
\newblock ``Quantum signal processing and nonlinear {{Fourier}} analysis''.
\newblock \href{https://dx.doi.org/10.1007/s13163-024-00494-5}{Revista Matem\'atica Complutense {\bf 37}, 655--694}~(2024).

\bibitem{rechtNecessarySufficientConditions2008}
Benjamin Recht, Weiyu Xu, and Babak Hassibi.
\newblock ``Necessary and {{Sufficient Conditions}} for {{Success}} of the {{Nuclear Norm Heuristic}} for {{Rank Minimization}}''.
\newblock In 2008 47th {{IEEE Conference}} on {{Decision}} and {{Control}}.
\newblock \href{https://dx.doi.org/10.1109/CDC.2008.4739332}{Pages 3065--3070}.
\newblock ~(2008).

\bibitem{dongInfiniteQuantumSignal2024}
Yulong Dong, Lin Lin, Hongkang Ni, and Jiasu Wang.
\newblock ``Infinite quantum signal processing''.
\newblock \href{https://dx.doi.org/10.22331/q-2024-12-10-1558}{Quantum {\bf 8}, 1558}~(2024).

\bibitem{dongGroundStatePreparationEnergy2022}
Yulong Dong, Lin Lin, and Yu~Tong.
\newblock ``Ground-{{State Preparation}} and {{Energy Estimation}} on {{Early Fault-Tolerant Quantum Computers}} via {{Quantum Eigenvalue Transformation}} of {{Unitary Matrices}}''.
\newblock \href{https://dx.doi.org/10.1103/PRXQuantum.3.040305}{PRX Quantum {\bf 3}, 40305}~(2022).

\bibitem{conwayCourseFunctionalAnalysis2007}
John~B. Conway.
\newblock ``A {{Course}} in {{Functional Analysis}}''.
\newblock \href{https://dx.doi.org/10.1007/978-1-4757-4383-8}{Volume~96 of Graduate {{Texts}} in {{Mathematics}}}.
\newblock Springer. New York, NY~(2007).

\bibitem{tsaiSU2NonlinearFourier2005}
Ya-Ju Tsai.
\newblock ``{{SU}}(2) non-linear {{Fourier}} transform''.
\newblock PhD thesis.
\newblock University of California Los Angeles.
\newblock ~(2005).

\end{thebibliography}

\pagebreak
\appendix

\section{Results from linear operators over Hilbert spaces}
\label{apx:trace-class}

\subsection{The $L^2$ space and linear operators}
{Here we briefly introduce the theory of linear operators over $L^2(\T^m, \calH)$ needed to define the adversary bound over square-integrable functions.

\begin{definition}
    The space $L^2(\T^m)$ is the Hilbert space of $m$-variable square-integrable functions on the unit torus $\T^m$, i.e., $f \in L^2(\T^m)$ if it satisfies
    \begin{align*}
        \norm{f}_{L^2} := \int_{\T^m} \abs{f(\vv{z})}^2 \ \dd\mu(\vv{z}) := \int_0^{2\pi} \cdots \int_0^{2\pi} \abs{f(e^{i\theta_1}, \ldots, e^{i\theta_m})}^2 \ \dd\theta_1 \cdots \dd\theta_m < \infty
    \end{align*}
    where $\mu$ is the normalized Lebesgue measure on $\T^m$.
\end{definition}
For a Hilbert space $\calH$, the space $L^2(\T^m, \calH)$ is the space of functions in the unit torus $\T^m$ with values in $\calH$ satisfying
\begin{align*}
    \norm{f}_{L^2} := \int_{\T^m} \norm{f(\vv{z})}_{\calH}^2 \ \dd\mu(\vv{z}) < \infty
\end{align*}
where $\norm{\cdot}_{\calH}$ is the norm induced by the inner product of $\calH$.}

{\begin{definition}
    A Hilbert space $\calH$ is said to be \emph{separable} if it admits a complete basis $\{ e_k \}_k$ of countable cardinality.
\end{definition}
We have that $L^2(\T^m)$ has a basis $\{ t_{\vv{k}} \}_{\vv{k} \in \Z^m}$ with $$t_{\vv{k}}(\vv{z}) = z_1^{k_1} \cdots z_m^{k_m} =: \vv{z}^{\vv{k}}$$ which is the Fourier basis. We will always assume the $\calH$ above to be separable, so the space $L^2(\T^m, \calH) \simeq L^2(\T^m) \otimes \calH$ will be separable as well since it has a basis $\{ t_{\vv{k}} \otimes g_h \}_{\vv{k},h}$, where $\{ g_h \}_h$ is a basis for $\calH$.

A linear operator $A$ acting on $L^2(\T^m, \calH)$ can be represented as an integral operator
\begin{align*}
    f \mapsto Af = \int_{\T^m} a(\cdot, \vv{y}) f(\vv{y}) \ \dd\mu(\vv{y})
\end{align*}
where $f \in L^2(\T^m, \calH)$, $f(\vv{y}) \in \calH$ and $a(\vv{x}, \vv{y})$ is a linear operator on $\calH$. Here $a$ is regarded as the integral kernel of $A$. To fix the ideas, if $\calH = \C^K$, then $a(\vv{x}, \vv{y})$ is a $K \times K$ matrix, and $f(\vv{y})$ is a vector of dimension $K$.

An operator $A$ on $\calH$ is said to be \emph{bounded} if, for some $M > 0$, $\norm{A f} \le M \norm{f}$ for any $f \in \calH$ (i.e., its operator norm $\norm{A}$ is finite). The \emph{rank} is the dimension of its range, and its \emph{adjoint} is defined as the unique operator $A^*$ satisfying $\langle A^* f, g \rangle = \langle f, A g \rangle$. An operator is \emph{Hermitian} or \emph{self-adjoint} if $A = A^*$, and \emph{positive semi-definite} $A \succeq 0$ if $\langle f, A f \rangle \ge 0$ for every $f \in \calH$. The concepts of $A \preceq 0$ and $A \preceq B$ are defined analogously.}

\subsection{Trace-class and Gram operators}
{Fixed some countable orthonormal basis $\{ e_k \}_k$ for the Hilbert space we can define a trace for a self-adjoint $A$:
\begin{align*}
    \Tr(A) = \sum_{k = 0}^\infty \langle e_k, A e_k \rangle
\end{align*}
and $A$ is said to be \emph{trace-class} if $\Tr \abs{A} := \Tr\sqrt{A^2} < \infty$\footnote{To give extra intuition, an operator is in the trace class when the sum of its eigenvalues converges absolutely, i.e., $\sum_{k = 0}^\infty |\lambda_n| < \infty$. This implies that the $\sum_{k = 0}^\infty \lambda_n$ converges as well, and regardless of how eigenvalues are arranged, which is necessary for the trace to be basis-independent.}. Note that any operator with finite rank is also trace-class. Moreover, if $A, B$ are trace-class, then also their linear combinations will be trace-class, because $\norm{A}_1 := \Tr\abs{A}$ is a norm. If an operator is trace-class, then the spectral theorem guarantees that it admits a complete (countable) orthonormal basis of eigenvectors and, in particular, there exists a sequence of eigenvalues~\cite{conwayCourseFunctionalAnalysis2007}.

The notion of trace-class operators is central in quantum mechanics, as density operators needed to describe open quantum systems of continuous variables are trace-class operators with unit trace. A stronger condition is when $A$ has finite rank, as in this case the eigenvalues are finitely many, and $A$ can be represented as a finite matrix after an appropriate change of basis.}

{\begin{definition}
    Let $v \in L^2(\T^m, \calW)$ with $\calW$ being separable. The \emph{Gram operator} $G_v$ of $v$ is the operator on $L^2(\T^m)$ with kernel
    \begin{align*}
        g_v(\vv{x}, \vv{y}) = \langle v(\vv{x}), v(\vv{y}) \rangle
    \end{align*}
    where the inner product is intended to be the one on $\calW$.
\end{definition}
Any Gram operator is positive semi-definite, moreover we can deduce the following:
\begin{theorem}
    \label{thm:gram-operator-is-trace-class}
    Any Gram operator is trace-class.
\end{theorem}
\begin{proof}
    By the Cauchy-Schwarz inequality on $\calW$, $|g_v(\vv{x}, \vv{y})|^2 \le \norm{v(\vv{x})}^2 \norm{v(\vv{y})}^2$. Therefore
    \begin{align*}
        \int_{\T^m} \int_{\T^m} |g_v(\vv{x}, \vv{y})|^2 \le \int_{\T^m} \norm{v(\vv{x})}^2 \int_{\T^m} \norm{v(\vv{y})}^2 < \infty
    \end{align*}
    implying $g_v \in L^2((\T^m)^2)$. Therefore, $G_v$ is a Hilbert-Schmidt operator, and thus the spectral theorem guarantees the existence of an orthonormal basis with a sequence of eigenvalues $\lambda_n$ that is square-summable~\cite{conwayCourseFunctionalAnalysis2007}\footnote{The spectral theorem generally holds for \emph{compact operators}. It is true that finite rank $\subseteq$ trace-class $\subseteq$ Hilbert-Schmidt $\subseteq$ compact~\cite{conwayCourseFunctionalAnalysis2007}.}. By the fact that $v \in L^2$, we can also conclude that
    \begin{align*}
        \int_{\T^m} g_v(\vv{x}, \vv{x}) < \infty \ .
    \end{align*}
    This is sufficient for $G_v$ to be trace-class, through an application of Parseval's identity~\cite{conwayCourseFunctionalAnalysis2007}.
\end{proof}
\begin{theorem}
    \label{thm:gram-operator-has-finite-rank}
    If $v \in L^2(\T^m, \C^K)$, then $G_v$ has rank $ \le K$.
\end{theorem}
\begin{proof}
    If $f \in L^2(\T^m)$, then
    \begin{align*}
        (G_v f)(\vv{x}) & = \int_{\T^m} \langle v(\vv{x}), v(\vv{y}) \rangle f(\vv{y}) \ \dd\mu(\vv{y}) \\
        & = \left\langle v(\vv{x}), \int_{\T^m} f(\vv{y}) v(\vv{y}) \ \dd\mu(\vv{y}) \right\rangle
    \end{align*}
    If we call $w(f) \in \C^K$ the vector on the right, then $G_v$ acts like:
    \begin{align*}
        f \mapsto \sum_{j = 1}^K w_j(f) \overline{v}_j
    \end{align*}
    i.e., the image of $G_v$ is contained in $\vspan{\overline{v}_1, \ldots, \overline{v}_K}$.
\end{proof}
As a corollary of Theorems~\ref{thm:gram-operator-is-trace-class},~\ref{thm:gram-operator-has-finite-rank}, the Gram operators for $G_{\xi}, G_{\tau}, G_v, G_{Ov}$ --- as well as their differences --- in Definition~\ref{def:adversary-bound-l2} are guaranteed to be trace-class. In particular, $G_{\xi}, G_{\tau}$ always have rank $\le K$, as $\xi(\vv{z}), \tau(\vv{z})$ have values in $\C^K$. In the case of polynomial state conversion, Theorem~\ref{thm:gamma-bound-finite-polynomial-multivariate} guarantees that also $G_v, G_{Ov}$ will also have finite rank, allowing us to represent everything as finite matrices.}

\subsection{Proof of weak duality}
{If $A$ is an operator with kernel $a$ and $B$ is an operator with kernel $b$, then the Hadamard product $A \circ B$ is defined as the operator with kernel $a b$ (note that this operation, as in the finite-dimensional case, is basis-dependent).

\begin{theorem}[Schur product theorem for operators]
    \label{thm:schur-product}
    Given two trace-class operators $A, B$ on some Hilbert space, if $A, B \succeq 0$ then $A \circ B \succeq 0$.
\end{theorem}
\begin{proof}
    Since $A, B$ are trace-class and self-adjoint, by the spectral theorem they can be written in terms of their eigendecompositions:
    \begin{align*}
        A & = \sum_{n = 0}^\infty \lambda_n \ketbra{u_n} \\
        B & = \sum_{m = 0}^\infty \mu_m \ketbra{w_m}
    \end{align*}
    The Hadamard product is then simply
    \begin{align*}
        A \circ B & = \sum_{n = 0}^\infty \sum_{m = 0}^\infty \lambda_n \mu_m \ketbra{u_n} \circ \ketbra{w_m}
    \end{align*}
    The integral kernel of the rank-1 operator $\ketbra{u_n}$ is $u_n(x) \overline{u_n(y)}$, and likewise for $w_n$, so the Hadamard product of these kernels is $u_n(x) w_m(x) \overline{u_n(y) w_m(y)}$, which is rank-1 and positive semi-definite. Therefore $A \circ B$ is a linear combination, with non-negative coefficients, of positive semi-definite operators, and therefore is positive semi-definite.
\end{proof}
Similarly, if $A \preceq 0$ and $B \succeq 0$ then $A \circ B \preceq 0$ by an analogous argument. Given a catalyst $v(\vv{z}) \in L^2(\T^m, \calM \otimes \calW)$, we can define its multiplication operator
\begin{align*}
    V : L^2(\T^m) & \rightarrow L^2(\T^m, \calM \otimes \calW) \\
    f(\vv{z}) & \mapsto f(\vv{z}) v(\vv{z})
\end{align*}
Constraint~(\ref{eq:unidir-gamma-bound-constraint-partial-l2}) can thus be rewritten as
\begin{align*}
    E \succeq V^* [(\Delta \otimes \id_{\calW})] V
\end{align*}
analogously to the finite-dimensional case. This is enough to extend Theorem~\ref{thm:relative-gamma-bound-weak-duality} to linear operators\footnote{The proof of Theorem~\ref{thm:relative-gamma-bound-weak-duality} proceeds as written. We should replace the maximum with a supremum, though.}. Also the proofs of Theorem~\ref{thm:gamma-bound-subnorm-property} and the lower bound of Theorem~\ref{thm:gamma-monte-carlo-bounds} carry over to this case without any modification, up to the change from discrete to continuous variable notation.
}
\section{Proofs for polynomial state conversion}
\label{apx:catalyst-proofs}

We give the detailed proofs of the results mentioned in Section~\ref{sec:state-conversion-l2}.

\subsection{Univariate case}
\label{apx:catalyst-proofs-univariate}

{
For a polynomial vector $p(z) = \sum_k p_k z^k$, we denote with $p_k$ its $k$-th Fourier coefficient, that is
\begin{align*}
    p_k = \langle t_k, p \rangle = \int_{\T} z^{-k} p(z) \ \dd\mu(z) = \frac{1}{2\pi} \int_0^{2\pi} e^{-ik\theta} p(e^{i\theta}) \ \dd\theta \ .
\end{align*}}

\begin{lemmarestate}{\ref{thm:feasible-space-polynomial-univariate}}
    If $v(z)$ is a feasible catalyst for a $\xi(z) \mapsto \tau(z)$ polynomial state conversion, where both $\xi, \tau$ have degree bounded by $n$, then $v(z)$ is a polynomial of degree at most $n - 1$.
\end{lemmarestate}
\begin{proof}
    The constraint of~(\ref{eq:unidir-gamma-bound-constraint-partial-l2}) can be rewritten in the Fourier basis:
    {\begin{align*}
        \langle t_{-k}, G_{\xi} t_{-h} \rangle & = \iint_{\T^2} x^k \langle \xi(x), \xi(y) \rangle y^{-h} \ \dd\mu(x) \dd\mu(y) \\
        & = \langle \int_{\T} x^{-k} \xi(x) \ \dd\mu(x), \int_{\T} y^{-h} \xi(y) \ \dd\mu(y) \rangle \\
        & = \langle \xi_k, \xi_h \rangle
    \end{align*}
    and similarly it goes for the Gram matrices of $\tau(z), v(z), z v(z)$ (the last one simply shifts the coefficients by $1$). Therefore, we have the following constraint:}
    \begin{align*}
        \langle \xi_k, \xi_h \rangle - \langle \tau_k, \tau_h \rangle = \langle v_k, v_h \rangle - \langle v_{k-1}, v_{h-1} \rangle
    \end{align*}
    for every $k, h \in \Z$. Assume for a contradiction that $v_{k'} \neq 0$ for some $k' < 0$, therefore by plugging $k = h = k'$ the constraint gives us
    \begin{align*}
        0 = \langle v_{k'}, v_{k'} \rangle - \langle v_{k'-1}, v_{k'-1} \rangle
    \end{align*}
    since $\xi, \tau$ do not have negative frequencies. By induction we thus have $\norm{v_k}^2 = \norm{v_{k'}}^2 \neq 0$ for every $k < k'$, which implies
    \begin{align*}
        \norm{v(z)}^2_{\infty} \ge \norm{v(z)}^2_{L^2} \ge \sum_{k = 1}^\infty \norm{v_{-k}}^2 = \sum_{k = 1}^\infty \norm{v_{k'}}^2 \ .
    \end{align*}
    The sum on the right-hand side diverges, and thus $v$ has infinite objective function, contradicting its feasibility. The middle inequality is a standard application of Bessel's inequality on the Hilbert space $L^2(\T, \calM \otimes \calW)$~\cite{conwayCourseFunctionalAnalysis2007}. Similarly, if $v_{k'} \neq 0$ for ${k'} \ge n$, then the constraint gives us $\norm{v_k}^2 = \norm{v_{k'}}^2$ for every $k > k'$, and the same argument on the $L^2$ norm applies.
\end{proof}
Put simply, the proof looks for a subsequence whose sum of squares diverges, in order to conclude that the full sum (which is equal to the $L^2$ norm, by Parseval's identity~\cite{conwayCourseFunctionalAnalysis2007}) diverges as well.

We now prove Theorems~\ref{thm:gamma-bound-to-gqsp},~\ref{thm:gamma-bound-to-qsp-sun}. We will prove the latter, as the former is a special case.
{
\begin{theoremrestate}{\ref{thm:gamma-bound-to-qsp-sun}}
    Let $r \ge 1$, and let $P, Q_1, \ldots, Q_r$ be polynomials of degree $\le n$ satisfying $|P(z)|^2 + \sum_{j=1}^r |Q_j(z)|^2 = 1$ on $\T$. There exists a bijective mapping between the catalysts in triangular form of the adversary bound for the state conversion problem $$\ket{0} \mapsto \tau(z) = P(z) \ket{0} + \sum_{j = 1}^r Q_j(z) \ket{j}$$ with oracle $O(z) = z$ and the QSP protocols over $SU(r+1)$ implementing $(P, Q_1, \ldots, Q_r)$.
\end{theoremrestate}
}
\begin{proof}
    {
    Let $A = (A_0, A_1, \ldots, A_n)$ be the sequence of processing operators of a QSP protocol over $SU(r + 1)$ implementing $\tau = (P, Q) := (P, Q_1, \ldots, Q_r)$, i.e.,
    \begin{align*}
        A_n \Tilde{w}^{(r)} \cdots \Tilde{w}^{(r)} A_0 \ket{0} = (P, Q) \ .
    \end{align*}
    Let $(v^{(k)}, s^{(k)}) := A_k \Tilde{w}^{(r)} \cdots \Tilde{w}^{(r)} A_0 \ket{0}$ be the polynomials at the $k$-th step, where $v^{(k)}$ is a scalar and $s^{(k)}$ has dimension $r$ (notice that for $k = n$ we get $(P, Q)$). We claim that the direct sum $v = v^{(0)} \oplus \cdots \oplus v^{(n-1)}$ is a valid catalyst in triangular form. For $0 < k \le n$, the QSP step $(v^{(k)}, s^{(k)}) = A_k (zv^{(k-1)}, s^{(k-1)})$ implies
    \begin{align*}
        \langle v^{(k)}(x), v^{(k)}(y) \rangle + \langle s^{(k)}(x), s^{(k)}(y) \rangle = \langle xv^{(k-1)}(x), yv^{(k-1)}(y) \rangle + \langle s^{(k-1)}(x), s^{(k-1)}(y) \rangle
    \end{align*}
    By summing over $k$ we get the condition:
    \begin{align*}
        \sum_{k = 1}^{n} \langle v^{(k)}(x), v^{(k)}(y) \rangle - \sum_{k = 0}^{n-1} \langle xv^{(k)}(x), yv^{(k)}(y) \rangle = \langle s^{(0)}(x), s^{(0)}(y) \rangle - \langle s^{(n)}(x), s^{(n)}(y) \rangle
    \end{align*}
    where the $s^{(k)}$ for the intermediate values of $k$ canceled out, and we brought the remaining terms on the other side of the equation. By knowing that $v^{(n)} = P$ and $s^{(n)} = Q$ (and so $(v^{(n)}, s^{(n)}) = \tau$), we can simply rewrite the equation as
    \begin{align*}
        \sum_{k = 1}^{n-1} \langle v^{(k)}(x), v^{(k)}(y) \rangle - \sum_{k = 0}^{n-1} \langle xv^{(k)}(x), yv^{(k)}(y) \rangle = \langle s^{(0)}(x), s^{(0)}(y) \rangle - \langle \tau(x), \tau(y) \rangle
    \end{align*}
    By adding the inner product for $v^{(0)}$ on both sides (and knowing that $(v^{(0)}, s^{(0)}) = A_0 \ket{0}$, which is constant), we obtain
    \begin{align}
        \langle v(x), v(y) \rangle - \langle xv(x), y(y) \rangle = 1 - \langle \tau(x), \tau(y) \rangle \label{eq:constraint-univariate-to-qsp-1}
    \end{align}
    which is exactly~(\ref{eq:unidir-gamma-bound-constraint-partial-l2}).}

    Conversely, consider a catalyst $v = v^{(0)} \oplus \cdots \oplus v^{(n-1)}$, which thus satisfies~(\ref{eq:constraint-univariate-to-qsp-1}). We proceed by induction on $n$: if $n = 0$ then Lemma~\ref{thm:feasible-space-polynomial-univariate} shows that $v = 0$ and thus by~(\ref{eq:constraint-univariate-to-qsp-1}) $\tau(z)$ is equal to $\ket{0}$ up to a unitary $A_0 \in SU(r+1)$. 
    
    When $n > 0$, the constraint of (\ref{eq:constraint-univariate-to-qsp-1}) can be rewritten as
    \begin{align}
        0 & = \langle \tau(x), \tau(y) \rangle - \langle x v^{(n-1)}(x), y v^{(n-1)}(y) \rangle \label{eq:constraint-univariate-to-qsp-2} \\
        & + \langle v^{(n-1)}(x), v^{(n-1)}(y) \rangle - \langle x v^{(n-2)}(x), y v^{(n-2)}(y) \rangle \nonumber \\
        & + \langle v^{(n-2)}(x), v^{(n-2)}(y) \rangle - \langle x v^{(n-3)}(x), y v^{(n-3)}(y) \rangle\nonumber \\
        &\ \ \ \ \ \vdots\nonumber \\
        & + \langle v^{(1)}(x), v^{(1)}(y) \rangle - \langle x v^{(0)}(x), y v^{(0)}(y) \rangle\nonumber \\
        & + \langle v^{(0)}(x), v^{(0)}(y) \rangle - 1\nonumber
    \end{align}
    By $\langle \tau(z), \tau(z) \rangle = 1$, which must hold as a polynomial equation, we get $\langle \tau_0, \tau_n \rangle = 0$ from the $n$-th degree term. {Define $\Pi_+$ to be the projector onto $\vspan{\tau_n}$, and let $\Pi_- := \id - \Pi_-$. From this, define $\tau_{\pm} = \Pi_{\pm} \tau$, and so $\tau = \tau_+ \oplus \tau_-$ up to a unitary transformation $A \in SU(r+1)$, (which will be used to construct the last processing operator in the QSP protocol)}. By construction, $\tau_+(z)$ has only degrees in $\{ 1, \ldots, n \}$, and $\tau_-(z)$ is supported on $\{ 0, \ldots, n-1 \}$, {so when we undo $\Tilde{w}^{(r)}$ to strip a layer
    \begin{align*}
        \tau \stackrel{A^\dag}{\mapsto} \tau_+ \oplus \tau_- \stackrel{(\Tilde{w}^{(r)})^\dag}{\mapsto} z^{-1} \tau_+ \oplus \tau_- \ .
    \end{align*}
    }We claim that $$\langle \tau_+(x), \tau_+(y) \rangle = \langle x v^{(n-1)}(x), y v^{(n-1)}(y) \rangle \ .$$ This would complete the proof by induction because {it would imply $z^{-1} \tau_+(z) = v^{(n-1)}(z)$ up to a phase} and, by (\ref{eq:constraint-univariate-to-qsp-2}) the remainder $v' = v^{(0)} \oplus \cdots \oplus v^{(n-2)}$ would be a catalyst for the state conversion $\ket{0} \mapsto v^{(n-1)} \oplus \tau_-$, a $(n-1)$-degree polynomial state\footnote{{Note that $v^{(n-1)}$ and $\tau_+$ are both scalar polynomials, so the two inner products being equal means that $\tau_+(z) = zv^{(n-1)}(z)$ up to a phase.}}.
    
    By~(\ref{eq:constraint-univariate-to-qsp-2}) we have that the two polynomial vectors
    \begin{align*}
        a(z) & := \tau_+(z) \ket{0}_A \ket{0}_B + \tau_-(z) \otimes \ket{0}_A \ket{1}_B + \sum_{k=0}^{n-1} v^{(k)}(z) \ket{1}_A \ket{k}_B \\
        b(z) & := zv^{(n-1)}(z) \ket{0}_A \ket{0}_B + \ket{0}_A \ket{1}_B + \sum_{k=1}^{n-1} zv^{(k-1)}(z) \ket{1}_A \ket{k}_B
    \end{align*}
    are equal up to a unitary transformation $U$. No polynomial, except for $\tau_+, zv^{(n-1)}$, has a $n$-degree term, hence we can extract the $n$-th Fourier coefficient
    \begin{align*}
        \tau_{+,n}\ket{0}_A \ket{0}_B = \frac{1}{2\pi} \int_0^{2\pi} a(e^{i\theta}) e^{-in\theta} d\theta = U \frac{1}{2\pi} \int_0^{2\pi} b(e^{i\theta}) e^{-in\theta} d\theta = U \qty( v^{(n-1)}_{n-1} \ket{0}_A \ket{0}_B ) \ .
    \end{align*}
    This implies that $\ket{0}_A \ket{0}_B$ is an eigenstate of $U$, and that $|v^{(n-1)}_{n-1}| = |\tau_{+,n}| \neq 0$ by unitarity\footnote{{$\tau_{+,n} \neq 0$ as we are assuming $\tau$ has degree exactly $n$. Otherwise also $v^{(n-1)} \equiv 0$ by Lemma~\ref{thm:feasible-space-polynomial-univariate}}.}. Let $e^{i\phi} = \tau_{+,n}/v^{(n-1)}_{n-1}$ be the associated eigenvalue. We can deduce that
    \begin{align*}
        \tau_+(z) = \bra{0}_A \bra{0}_B \ket{a(z)} = \bra{0}_A \bra{0}_B U \ket{b(z)} = e^{-i\phi} \bra{0}_A \bra{0}_B \ket{b(z)} = e^{-i\phi} z v^{(n-1)}(z)
    \end{align*}
    which implies $\langle \tau_+(x), \tau_+(y) \rangle = \langle x v^{(n-1)}(x), y v^{(n-1)}(y) \rangle$ as claimed.

    We remark once again that this can be seen as a layer stripping argument, {as we are carrying out the chain of transformations}
    \begin{align*}
        \tau \overset{A^\dag}{\mapsto} \tau_+ \oplus \tau_- = e^{-i\phi} z v^{(n-1)} \oplus \tau_- \mapsto z v^{(n-1)} \oplus \tau_- \overset{(\Tilde{w}^{(r)})^\dag}{\mapsto} v^{(n-1)} \oplus \tau_-
    \end{align*}
    and the last processing operator $A_n^\dag$ for the QSP protocol is the composition of the first two transformations.
\end{proof}
{From this argument we can better understand the structure of the catalysts $v$ for this specific case of state conversion.}
\begin{corollary}
    If $v(z) = v^{(0)} \oplus \cdots \oplus v^{(n-1)}$ {is in triangular form}, then $v^{(k)}$ has exactly degree $k$ for every $k$.
\end{corollary}
Choosing a $QL$-decomposition with positive diagonal for the triangular form (which is unique) is equivalent to the constraint $a^*(0) > 0$ that guarantees uniqueness in the literature for the non-linear Fourier transform~\cite{tsaiSU2NonlinearFourier2005,alexisQuantumSignalProcessing2024,alexisInfiniteQuantumSignal2024}.

\subsection{Multivariate case}
\label{apx:catalyst-proofs-multivariate}

{Much like the univariate case, for a polynomial vector $p(z) = \sum_{\vv{k} \in \Z^m} p_{\vv{k}} \vv{z}^{\vv{k}}$, we denote with $p_{\vv{k}}$ its $\vv{k}$-th Fourier coefficient, that is
\begin{align*}
    p_{\vv{k}} = \langle t_{\vv{k}}, p \rangle = \iint_{\T^{m}} \vv{z}^{-\vv{k}} f(\vv{z}) \ \dd\mu(\vv{z}) \ .
\end{align*}
Before proving the results, we show how to obtain~(\ref{eq:gamma-bound-constraint-multivariate-fourier}): the argument for the Gram operators of $\xi, \tau, v$ do not change from the univariate case. On the other hand, the Gram operator for $Ov$ has kernel:
\begin{align*}
    \langle O(\vv{x}) v(\vv{x}), O(\vv{y}) v(\vv{y}) \rangle & = \sum_{i = 1}^m \sum_{j = 1}^m \langle x_i \Pi_i v(\vv{x}), y_j \Pi_j v(\vv{y}) \rangle \\
    & = \sum_{j = 1}^m \langle x_j \Pi_j v(\vv{x}), y_j \Pi_j v(\vv{y}) \rangle & \text{since $\Pi_i \Pi_j = 0$ for $i \neq j$}
\end{align*}
We thus consider the Gram operators $G_{z_j \Pi_j v}$ separately for each $j$:
\begin{align*}
    \langle t_{-\vv{k}}, G_{Ov} t_{-\vv{h}} \rangle & = \int_{\T^m} \int_{\T^m} \vv{x}^{\vv{k}} \langle x_j \Pi_j v(\vv{x}), y_j \Pi_j v(\vv{y}) \rangle \vv{y}^{-\vv{h}} \ \dd\mu(\vv{x}) \dd\mu(\vv{y}) \\
    & = \langle \Pi_j \int_{\T^m} x_j \vv{x}^{-\vv{k}} v(\vv{x}) \ \dd\mu(\vv{x}), \Pi_j \int_{\T^m} y_j \vv{y}^{-{\vv{h}}} v(\vv{y}) \ \dd\mu(\vv{y}) \rangle \\
    & = \langle v_{\vv{k} - \vv{e}_j}, \Pi_j v_{\vv{h} - \vv{e}_j} \rangle
\end{align*}
where we simply notice that $x_j \vv{x}^{-\vv{k}} = \vv{x}^{-(\vv{k} - \vv{e}_j)}$, and we take advantage of the fact that $\Pi_j$ is idempotent and Hermitian.}

\begin{lemmarestate}{\ref{thm:gamma-bound-finite-polynomial-multivariate}}
    Suppose the oracle is $O(\vv{z}) = \diag(\vv{z})$ and that $\xi(\vv{z}), \tau(\vv{z})$ have degree $\le n_j$ in $z_j$ for every $j$. Then any feasible solution $v(\vv{z})$ to the total polynomial state conversion $\xi \mapsto \tau$ must be a polynomial of degree $\le n_j - 1$ in each variable $z_j$.
\end{lemmarestate}
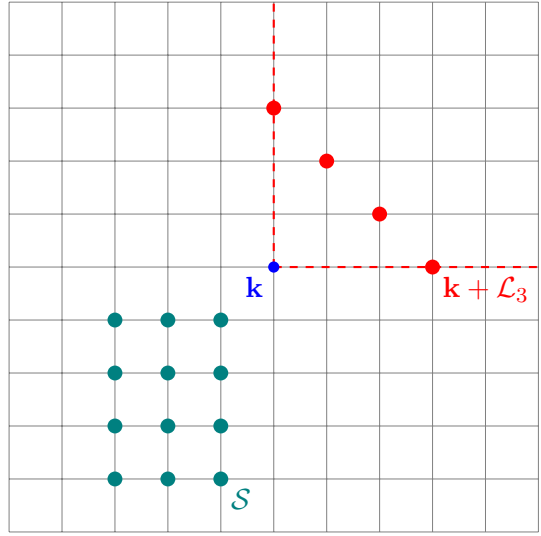
\begin{figure}
    \centering
    \begin{tikzpicture}[scale=0.7]
    \draw[step=1cm,gray,very thin] (0,0) grid (10,10);

    \def\kx{5}
    \def\ky{5}
  
    \foreach \i/\j in {8/5, 7/6, 6/7, 5/8} {
      \fill[red] (\i,\j) circle (4pt);
    }

    \foreach \i in {2,...,4} {
        \foreach \j in {1,...,4}
        \fill[teal] (\i,\j) circle (4pt);
    }
  
    \draw[red, dashed, thick] (\kx, \ky) -- ({\kx+5},\ky);
    \draw[red, dashed, thick] (\kx, \ky) -- (\kx, {\ky+5});

    \fill[blue] (\kx,\ky) circle (3pt);
    \node[below left, blue] at (\kx,\ky) {$\vv{k}$};
    \node[below right, red] at ({\kx+3},\ky) {$\vv{k} + \calL_3$};
    \node[below right, teal] at (4,1) {$\calS$};

  \end{tikzpicture}
    \caption{Visualization of the argument in two variables. The set $\vv{k} + \calL_3$ in red. The region on the top-right of the dashed red lines is the one constituting the divergent subsequence.}
    \label{fig:state-2d-argument}
\end{figure}
\begin{proof}
    The idea is to look for a divergent sum as in Lemma~\ref{thm:feasible-space-polynomial-univariate}. We claim that the hyper-rectangle $$\calS = \{ \vv{h} \in \Z^m : 0 \le h_j < n_j \}$$ contains the support of $v$. For this purpose, suppose for a contradiction that $v_{\vv{k}} \neq 0$ for some $\vv{k} \not\in \calS$. For $s \ge 0$, we define the subset
    \begin{align*}
        \calL_s & = \qty{ \vv{h} \in \Z^m_{\ge 0} : \sum_j h_j = s }
    \end{align*}
    One of the following must hold:
    \begin{enumerate}[(i)]
        \item $(\vv{k} + \calL_s) \cap \calS = \emptyset$ for every $s \in \Z_{\ge 0}$ or
        \item $(\vv{k} - \calL_s) \cap \calS = \emptyset$ for every $s \in \Z_{\ge 0}$
    \end{enumerate}
    Suppose for a contradiction that both claims are false: this implies that $\vv{k} + \vv{d}_1, \vv{k} - \vv{d}_2 \in \calS$ for some $\vv{d}_1, \vv{d}_2 \ge 0$. These two points identify a hyper-rectangle contained in $\calS$, and since we have that, for every $j$
    \begin{align*}
        0 \le (\vv{k} - \vv{d}_2)_j \le k_j \le (\vv{k} + \vv{d}_1)_j < n_j
    \end{align*}
    we would conclude $\vv{k} \in \calS$, a contradiction (see Figure~\ref{fig:state-2d-argument} for a visualization on the case $m = 2$). If (i) holds then, for any $s \in \Z_{> 0}$ and $\vv{h} \in \vv{k} + \calL_s$, (\ref{eq:gamma-bound-constraint-multivariate-fourier}) becomes
    \begin{align*}
        \langle v_{\vv{h}}, v_{\vv{h}} \rangle = \sum_{j = 1}^m \langle v_{\vv{h} - \vv{e}_j}, \Pi_j v_{\vv{h} - \vv{e}_j} \rangle \ .
    \end{align*}
    By summing over all $\vv{h} \in \vv{k} + \calL_1$, we obtain
    \begin{align*}
        \sum_{\vv{h} \in \vv{k} + \calL_1}\langle v_{\vv{h}}, v_{\vv{h}} \rangle & = \sum_{\vv{h} \in \vv{k} + \calL_1} \sum_{j = 1}^m \langle v_{\vv{h} - \vv{e}_j}, \Pi_j v_{\vv{h} - \vv{e}_j} \rangle \ge \langle v_{\vv{k}}, v_{\vv{k}} \rangle
    \end{align*}
    where the inequality comes from the fact that the sum in the middle contains all the pieces of $v_{\vv{k}}$. One can replace $\calL_1$ with $\calL_s$, by an induction argument on $s$ (more precisely, the sum over $\vv{k} + \calL_s$ contains all the pieces of the elements in $\vv{k} + \calL_{s-1}$). Since $v_{\vv{k}} \neq 0$, the non-decreasing sequence $$a_s := \sum_{\vv{h} \in \vv{k} + \calL_s} \langle v_{\vv{h}}, v_{\vv{h}} \rangle$$ does not converge to $0$, contradicting the fact that $v \in L^2$. Analogously if (ii) holds, on the other hand, \begin{align*}
        \langle v_{\vv{k}}, v_{\vv{k}} \rangle = \sum_{j = 1}^m \langle v_{\vv{k} - \vv{e}_j}, \Pi_j v_{\vv{k} - \vv{e}_j} \rangle \le \sum_{\vv{h} \in \vv{k} - \calL_1} \langle v_{\vv{h}}, v_{\vv{h}} \rangle \ .
    \end{align*}
    By an induction on $s$, we obtain that the sequence $$b_s := \sum_{\vv{h} \in \vv{k} - \calL_s} \langle v_{\vv{h}}, v_{\vv{h}} \rangle$$
    does not converge to $0$, contradicting once again $v \in L^2$.
\end{proof}

\begin{theoremrestate}{\ref{thm:gamma-bound-to-mqsp}}
    Consider the total state conversion $$\ket{0} \mapsto \tau(\vv{z}) = P(\vv{z}) \ket{0} + \sum_{k = 1}^r Q_k(\vv{z}) \ket{k} \ .$$
    where $P, Q_k$ have degree $\le n$. {A catalyst in triangular form of the adversary bound for the above state conversion problem can be turned into a M-QSP protocol of length $n$ for $\tau(\vv{z})$.}
\end{theoremrestate}
\begin{proof}
    We proceed by induction on $n$, the case $n = 0$ being trivial. For $n > 0$, let $\tau(\vv{z}) = (P, Q_1, \ldots, Q_k) =: (P, Q)$. {Let $\Pi_+$ be the projector onto the span of all the $n$-degree coefficients of $\tau$
    \begin{align*}
        \vspan{ \tau_{\vv{k}} : \norm{\vv{k}}_1 = n } \ .
    \end{align*}
    Define $\Pi_- := \id - \Pi_+$ and let $\tau_{\pm} := \Pi_{\pm} \tau$. Then $\tau = \tau_+ \oplus \tau_-$ up to a unitary transformation $A$,} where $\tau_+$ is a component of degree $n$ and $\tau_-$ is a component of degree $\le n - 1$. Intuitively, $\tau_-$ will be the component multiplied by $1$ by the inverse signal operator, just like in Theorem~\ref{thm:gamma-bound-to-gqsp}. {We remark that in this case $\tau_+$ need not be scalar, and different dimensions can be multiplied by different variables.} We conclude the proof by showing that
    \begin{align}
        \langle \tau_+(\vv{x}), \tau_+(\vv{y}) \rangle = \sum_{j=1}^m \langle x_j \Pi_j v^{(n-1)}(\vv{x}), y_j \Pi_j v^{(n-1)}(\vv{y}) \rangle = \langle O v^{(n-1)}(\vv{x}), O v^{(n-1)}(\vv{y}) \rangle \ . \label{eq:gamma-bound-to-mqsp-gram-equality}
    \end{align}
    This would imply that there exists a unitary mapping $T$ satisfying
    \begin{align*}
        T Ov^{(n-1)} = T \sum_j z_j \Pi_j v^{(n-1)}(\vv{z}) = \tau_+(\vv{z}) \ .
    \end{align*}
    Hence, the processing operator $A_n = A(T \oplus \id)$ would then carry out the transformation
    \begin{align*}
        v^{(n-1)} \oplus \tau_- \stackrel{O}{\mapsto} O v^{(n-1)} \oplus \tau_- \stackrel{T}{\mapsto} \tau_+ \oplus \tau_- \stackrel{A}{\mapsto} \tau
    \end{align*}
    thus choosing the last processing operator to be $A_n = A(T \oplus \id)$ allowing us to reduce ourselves to the $(n-1)$-degree state conversion $\ket{0} \mapsto v^{(n-1)} \oplus \tau_-$, for which the remainder $v^{(0)} \oplus \cdots \oplus v^{(n-2)}$ is a valid catalyst.

    Similarly to Theorem~\ref{thm:gamma-bound-to-gqsp}, we know by (\ref{eq:unidir-gamma-bound-constraint-partial-l2}) that the Gram matrices of $\ket{0} \oplus O v(\vv{z})$ and $\tau_+(\vv{z}) \oplus \tau_-(\vv{z}) \oplus v(\vv{z})$ are equal. Thus, by rearranging the terms, the two vectors
    \begin{align*}
        a(\vv{z}) & := \tau_+(\vv{z}) \otimes \ket{0}_A \ket{0}_B + \tau_-(\vv{z}) \otimes \ket{0}_A \ket{1}_B + \sum_{k=0}^{n-1} v^{(k)}(\vv{z}) \otimes \ket{1}_A \ket{k}_B \\
        b(\vv{z}) & := O v^{(n-1)}(\vv{z}) \otimes \ket{0}_A \ket{0}_B + \ket{0}_A \ket{1}_B + \sum_{k=1}^{n-1} O v^{(k-1)}(\vv{z}) \otimes \ket{1}_A \ket{k}_B \ .
    \end{align*}
    are equal up to a unitary transformation $U$. The only difference with Theorem~\ref{thm:gamma-bound-to-gqsp} is that $\tau_+, \tau_-, v^{(k)}$ are not necessarily scalars. We introduce the following notation to extract degrees from a polynomial vector:
    \begin{align*}
        [p(\vv{z})]_{\vv{k}} := \int_{\T^m} p(\vv{z}) \vv{z}^{-\vv{k}} \ d\mu(\vv{z}) \ .
    \end{align*}
    The only terms of degree $n$ are $\tau_+, Ov^{(n-1)}$, hence for any $n$-degree term $\vv{z}^\vv{k}$, the following holds
    \begin{align*}
        [\tau_{+}]_{\vv{k}} \otimes \ket{0}_A \ket{0}_B = U [O v^{(n-1)}]_{\vv{k}} \otimes \ket{0}_A \ket{0}_B \ .
    \end{align*}
    Which means that $U$ keeps $\ket{0}_A \otimes \ket{0}_B$ fixed. In other words, $U$ can be represented as
    \begin{align*}
        U =
        \begin{bmatrix}
            T & 0 \\
            0 & \cdot
        \end{bmatrix}
    \end{align*}
    where the first row/column represents the subspace spanned by $\ket{0}_A \otimes \ket{0}_B$. By unitarity of $U$, also $T$ must be unitary. This shows that $\tau_+ = T O v^{(n-1)}$ as claimed.
\end{proof}

\end{document}